\documentclass[11pt]{article}
\usepackage[margin=1in]{geometry}
\usepackage[T1]{fontenc}
\usepackage[utf8]{inputenc}
\usepackage{microtype}
\usepackage[all=normal,bibliography=tight]{savetrees}

\usepackage{listings}
\usepackage{cite}
  \usepackage{mathrsfs}

  \usepackage{amstext,amsfonts,amsthm,amsmath,amssymb}

  \usepackage{tikz}
\usepackage{pgflibrarysnakes}
\usetikzlibrary{snakes}
\usetikzlibrary{calc}
\usetikzlibrary{decorations.markings}
  
\usepackage{graphicx}
\usepackage{comment}
\usepackage{url}
\usepackage{xspace}
\usepackage{todonotes}
\usepackage[shortlabels]{enumitem}
\setlist{topsep=1ex,itemsep=-1ex,partopsep=0ex,parsep=1ex}
\usepackage[absolute]{textpos}

\def\cqedsymbol{\ifmmode$\lrcorner$\else{\unskip\nobreak\hfil
\penalty50\hskip1em\null\nobreak\hfil$\lrcorner$
\parfillskip=0pt\finalhyphendemerits=0\endgraf}\fi} 

\newcommand{\cqed}{\renewcommand{\qed}{\cqedsymbol}}

\newtheorem{lemma}{Lemma}[section]

\newtheorem{theorem}[lemma]{Theorem}
\newtheorem{claim}[lemma]{Claim}
\newtheorem{conjecture}[lemma]{Conjecture}
\theoremstyle{definition}
\newtheorem{definition}[lemma]{Definition}

\newcommand{\Oh}{\mathcal{O}}
\newcommand{\cc}{\mathtt{cc}}
\newcommand{\poly}{\mathrm{poly}}
\newcommand{\eps}{\varepsilon}
\newcommand{\Qq}{\mathcal{Q}}
\newcommand{\Rr}{\mathcal{R}}

\newcommand{\N}{\mathbb{N}}

\newcommand{\Cc}{\mathcal{C}}
\newcommand{\Dd}{\mathcal{D}}
\newcommand{\Nn}{\mathcal{N}}

\renewcommand{\leq}{\leqslant}
\renewcommand{\geq}{\geqslant}

\newcommand{\esd}{\eta}

\newcommand{\wei}{\mathbf{w}}

\title{Quasi-polynomial time approximation schemes for the Maximum Weight Independent Set Problem in $H$-free graphs%
\thanks{An extended abstract of this work appeared at SODA 2020~\cite{ChudnovskyPPT20}.
The SODA 2020 version contained only approximation schemes, but not the subexponential algorithms presented here.}}

\author{ 
  Maria Chudnovsky\thanks{Supported by NSF grants DMS-1763817. This material is based upon work supported in part by the U. S. Army Research Office under
    grant number  W911NF-16-1-0404.}\\
Princeton University, Princeton, NJ 08544 \and
 Marcin Pilipczuk\thanks{This research is a part of a project that has received funding from the European Research Council (ERC)
under the European Union's Horizon 2020 research and innovation programme
Grant Agreement no.~714704.} 
  \\ Institute of Informatics, University of Warsaw\\Banacha 2, 02-097 Warsaw, Poland \and 
Micha\l{} Pilipczuk\thanks{This research is a part of a project that has received funding from the European Research Council (ERC)
under the European Union's Horizon 2020 research and innovation programme
Grant Agreement no.~677651.} 
  \\ Institute of Informatics, University of Warsaw\\Banacha 2, 02-097 Warsaw, Poland \and 
    St\'{e}phan Thomass\'{e} \\ 
   Institut Universitaire de France \\
   Laboratoire d'Informatique du Parallélisme, UMR 5668 ENS Lyon, \\ CNRS, UCBL, INRIA, Université de Lyon, France
}

\date{}

\begin{document}

\maketitle

\begin{textblock}{20}(0, 13.0)
\includegraphics[width=40px]{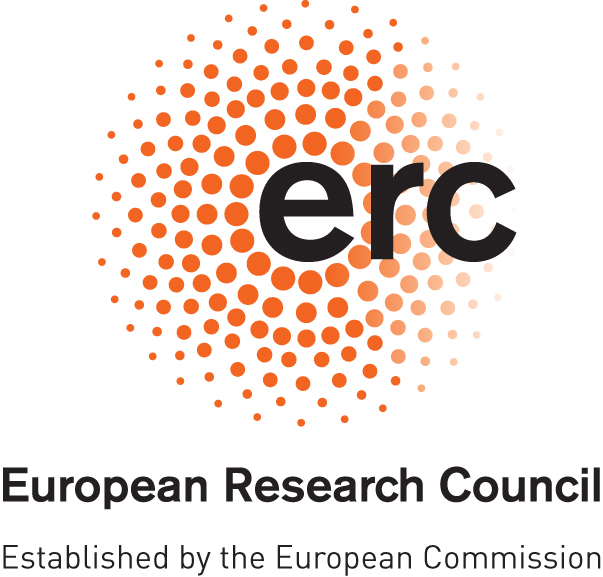}%
\end{textblock}
\begin{textblock}{20}(-0.25, 13.4)
\includegraphics[width=60px]{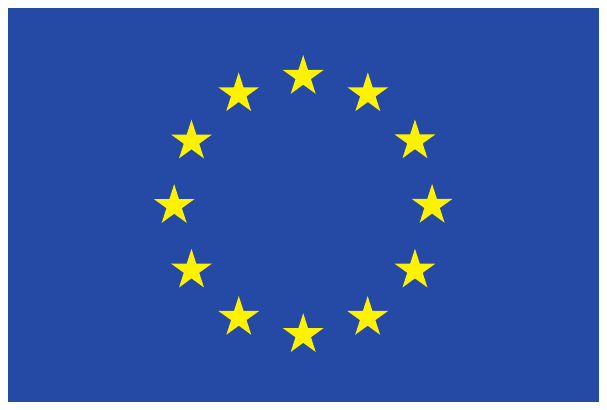}%
\end{textblock}

\begin{abstract}
In the \textsc{Maximum Independent Set} problem we are asked to find a set of pairwise nonadjacent vertices in a given graph with the maximum possible cardinality.
In general graphs, this classical problem is known to be NP-hard and hard to approximate within a factor of $n^{1-\varepsilon}$ for any $\varepsilon > 0$.
Due to this, investigating the complexity of \textsc{Maximum Independent Set} in various graph classes in hope of finding better tractability results is an active research direction.

In $H$-free graphs, that is, graphs not containing a fixed graph $H$ as an induced subgraph, the problem is known to remain NP-hard and APX-hard
whenever $H$ contains a cycle, a vertex of degree at least four, or two vertices of degree at least three in one connected component. 
For the remaining cases, where every component of $H$ is a path or a subdivided claw, the complexity of \textsc{Maximum Independent Set} remains
widely open, with only a handful of polynomial-time solvability results for small graphs $H$ such as $P_5$, $P_6$, the claw, or the fork.

We prove that for every such ``possibly tractable'' graph $H$ there exists an algorithm that, given an $H$-free graph $G$ and an accuracy parameter $\varepsilon > 0$,
   finds an independent set in $G$ of cardinality within a factor of $(1-\varepsilon)$ of the optimum 
   in time exponential in a polynomial of $\log |V(G)|$ and $\varepsilon^{-1}$.
   Furthermore, an independent set of maximum size can be found in subexponential time $2^{\Oh(|V(G)|^{8/9} \log |V(G)|)}$.
   That is, we show that for every graph $H$ for which \textsc{Maximum Independent Set} is not known to be APX-hard and SUBEXP-hard in $H$-free graphs,
   the problem admits a quasi-polynomial time approximation scheme and a subexponential-time exact algorithm in this graph class.
   Our algorithms work also in the more general weighted setting, where the input graph is supplied with a weight function on vertices and we are maximizing the total weight of an independent set.

\end{abstract}

\section{Introduction}
For an undirected graph $G$, a vertex subset $X \subseteq V(G)$ is \emph{independent} if no two vertices of $X$ are adjacent.
The size of the largest independent set in a graph, often denoted as $\alpha(G)$, is one of the fundamental graph parameters studied in graph theory. 
Therefore, it is natural to study the computational task of computing $\alpha(G)$, given $G$, which we call the \textsc{Maximum Independent Set} problem (\textsc{MIS}).
In the weighted generalization, \textsc{Maximum Weight Independent Set} (\textsc{MWIS}), the given graph $G$ is supplied with a weight function $\wei \colon V(G) \to \mathbb{N}$,
and we ask for an independent set $X$ in $G$ with the maximum possible total weight $\wei(X) = \sum_{x \in X} \wei(x)$. 
\textsc{MIS} is a classic problem that is known not only to be NP-hard, but also hard to approximate within a factor of $n^{1-\varepsilon}$ for every $\varepsilon > 0$, unless $\mathrm{P}=\mathrm{NP}$~\cite{Hastad99,Zuckerman07}.

In light of these lower bounds, a lot of effort has been put into understanding the complexity of \textsc{MIS} and \textsc{MWIS} in restricted graph classes. 
While the celebrated Baker's technique yields a polynomial-time approximation scheme (PTAS) for \textsc{MWIS} in planar graphs~\cite{Baker94},
\textsc{MIS} remains NP-hard in planar graphs of degree at most three and APX-hard in graphs of maximum degree at most three~\cite{GJbook,GareyJ77,GuruswamiS11}.
To extend these lower bounds to other graph classes, the following observation due to Poljak~\cite{Poljak74} is very useful: if $G'$ is created from $G$ by subdividing one edge twice,
then $\alpha(G') = \alpha(G)+1$.
Thus, if we fix any graph $H$ that contains either a cycle, a vertex of degree at least four, or two vertices of degree three in one connected component, then 
starting from a graph $G$ of maximum degree at most three
and subdividing each edge a sufficient number of times, we obtain a graph $G'$ where computing $\alpha(\cdot)$ is equally hard, while $G$ does not contain an induced subgraph isomorphic to~$H$.
Now, \textsc{MIS} is known to be APX-hard in graphs of maximum degree at most three, and in this case $\alpha(G)$ is linear in the size of the graph.
Moreover, under the Exponential-Time Hypothesis, \textsc{MIS} has no {\em{subexponential-time algorithm}} (that is, one with running time $2^{o(n)}$)
on graphs of maximum degree at most three; we call this property {\em{SUBEXP-hardness}}.
This implies that \textsc{MIS} remains APX-hard and SUBEXP-hard in $\mathcal{H}$-free graphs for every finite family of graphs $\mathcal{H}$ such that every $H \in \mathcal{H}$ is not a disjoint union of paths and subdivided claws.%
\footnote{A graph is \emph{$H$-free} if it does not contain an induced subgraph isomorphic to $H$.
  A graph $G$ is $\mathcal{H}$-free if $G$ is $H$-free for every $H \in \mathcal{H}$.
  A \emph{subdivided claw} is a tree with one vertex of degree three and all other vertices of degree at most two.}
  
However, when $H$ is a disjoint union of paths and subdivided claws, no hardness result on the complexity of \textsc{MIS} nor \textsc{MWIS} on $H$-free graphs is known.
In fact, it would be consistent with our knowledge if \textsc{MWIS} turns out to be polynomial-time solvable in $H$-free graphs for all such graphs $H$. 
Currently, such a result seems beyond our reach.
Let $P_t$ be the path on $t$ vertices and the \emph{claw} be the four-vertex tree with one vertex of degree three and three leaves. 
The class of $P_4$-free graphs (known also as \emph{cographs}) have a very  rigid structure (in particular, they have clique-width at most $2$), 
and hence they admit a simple polynomial-time algorithm for \textsc{MWIS}~\cite{CorneilLB81}.
Claw-free graphs also possess very strong structural properties and inherit many properties of their main subclass: line graphs.
In particular, the augmenting-path algorithm for maximum cardinality matching generalizes to a polynomial-time algorithm for \textsc{MWIS} in claw-free graphs~\cite{Minty80,Sbihi80,NakamuraT98}.
A more modern approach based on the decomposition theorems for claw-free graphs yields a different
line of algorithms~\cite{FaenzaOS11,FaenzaOS14,NobiliS17,NobiliS15a,OrioloPS08}. 
This, in turn, can be generalized to so-called fork-free graphs~\cite{LozinM08}, where the fork is constructed from the claw by subdividing one edge once. 
The case of $P_5$-free graphs, after being open for a long time, was resolved positively in 2014 by Lokshtanov, Vatshelle, and Villanger~\cite{LokshtanovVV14}
using the framework of \emph{potential maximal cliques}. 
With a substantially larger technical effort, their approach has been generalized to $P_6$-free graphs by Grzesik et al.~\cite{GrzesikKPP22}.
The polynomial-time solvability of \textsc{MWIS} on $P_7$-free graphs, or $T$-free graphs where $T$ is any subdivision of the claw other than the fork, remains open.
There is a significant body of work concerning the complexity of \textsc{MWIS} in various subclasses of $P_t$-free or $T$-free graphs, see e.g.~\cite{p71,p67,p65,p73,p72}.

In 2017, evidence in favor of the tractability of \textsc{MIS} and \textsc{MWIS} at least in $P_t$-free graphs has been found: 
 there is a subexponential-time algorithm for the problem running in time $2^{\Oh(\sqrt{n t \log n})}$ on an $n$-vertex $P_t$-free graph~\cite{BacsoLMPTL19,Brause17,GORSSS18}.
The main insight is that the
classical Gy\'{a}rf\'{a}s' path argument, originally used to show that $P_t$-free graphs are $\chi$-bounded~\cite{gyarfas},
implies that a $P_t$-free graph $G$ admits a balanced separator
consisting of at most $t-1$ vertex neighborhoods. Here, a balanced separator is a set of vertices whose removal results in a graph where every connected component has at most $|V(G)|/2$ vertices.

\paragraph{Our results.}
We provide a new evidence in favor of the tractability of \textsc{MWIS} in all cases of $H$-free graphs where it is not known to be APX-hard.
\begin{theorem}\label{thm:main}
For every graph $H$ whose every connected component is a path or a subdivided claw, there exists an algorithm that,
given an $H$-free graph $G$ with a weight function $\wei : V(G) \to \mathbb{N}$ and an accuracy parameter $\varepsilon > 0$,
computes a $(1-\varepsilon)$-approximation to \textsc{Maximum Weight Independent Set} on $(G,\wei)$
in time exponential in a polynomial of $\log |V(G)|$ and $\varepsilon^{-1}$.
\end{theorem}
\begin{theorem}\label{thm:main2}
For every graph $H$ whose every connected component is a path or a subdivided claw, there exists an algorithm that,
given an $H$-free graph $G$ with a weight function $\wei : V(G) \to \mathbb{N}$,
solves \textsc{Maximum Weight Independent Set} on $(G,\wei)$
in time exponential in $\Oh(|V(G)|^{8/9} \log |V(G)|)$.
\end{theorem}
That is, in all the cases when \textsc{MWIS} is not known to be APX-hard or SUBEXP-hard on $H$-free graphs, we prove that \textsc{MWIS} admits a quasi-polynomial time approximation scheme (QPTAS)
  and a subexponential-time algorithm.

We remark here that Theorems~\ref{thm:main} and~\ref{thm:main2} treat $H$ as a constant-sized graph. That is, the polynomial (of $\log |V(G)|$ and $\varepsilon^{-1}$)
and the constant factor (hidden in the big-$\Oh$ notation)
in the exponents of the time bounds depend on the graph $H$. 
If one follows closely the arguments, the final bound of the running time of the approximation scheme is of the form $\exp(c_H \varepsilon^{-c} \log^c n)$ for some constant $c_H$ depending on $H$ and a universal (independent of $H$) constant $c$. Since the constant $c$ is significantly larger than $1$, we refrain from precisely analysing this running time bound for the sake of simplicity.

For an insight into the techniques standing behind Theorems~\ref{thm:main} and~\ref{thm:main2}, let us first focus on the case $H = P_t$. 
A subexponential-time algorithm for this case has been already provided in~\cite{BacsoLMPTL19}.
For an approximation scheme, let $(G,\wei)$ be an input to \textsc{MWIS} with $G$ being $P_t$-free and let $\varepsilon > 0$ be an accuracy parameter.
Let $X \subseteq V(G)$ be an independent set in $G$ of maximum possible weight.
Fix a threshold $\beta := \varepsilon^{-1} t \log n$ and say that a vertex $v \in V(G)$ is \emph{$X$-heavy} if it contains at least a $\beta^{-1}$ fraction of the weight of $X$
in its closed neighborhood, that is, $\wei(X \cap N[v]) \geq \beta^{-1} \wei(X)$. 
A simple coupon-collecting argument shows that there is a set $Y \subseteq X$ of size $\Oh(\beta \log n)$ such that all $X$-heavy vertices are contained in $N[Y]$. 
We investigate all the $n^{\Oh(\beta \log n)} = 2^{\Oh(\varepsilon^{-1} \log^3 n)}$ subcases corresponding to the possible choices of $Y$.
Having fixed $Y$ in a subcase, we can delete $N(Y)$ from the graph and from now on assume that there are no more $X$-heavy vertices (except for isolated vertices that are easy to deal with). 

Now the Gy\'{a}rf\'{a}s path argument, like e.g. in~\cite{GORSSS18}, asserts that in $G$ there exists a balanced separator $A = N[B]$ for some $|B| \leq t-1$. 
We simply delete $A$ from the graph and restart the whole algorithm on every connected component of $G$. Since there are no $X$-heavy vertices, 
we lose only a fraction of $\beta^{-1} t < \varepsilon / \log n$ of the weight of $X$ in this step. 
Since every connected component of $G-A$ is of size at most $n/2$, the depth of the recursion is at most $\log n$.
Consequently, throughout the recursion the total loss in the weight of the optimum solution $X$ is at most $\varepsilon\cdot \wei(X)$.
Furthermore, it can be easily seen that the whole recursion tree has size bounded by $2^{\Oh(\varepsilon^{-1} \log^4 n)}$, giving a quasi-polynomial running time bound of the whole algorithm.

To generalize this argument (and the argument for subexponential-time algorithm of~\cite{BacsoLMPTL19}) 
to the case of $H$ being a subdivided claw, an additional ingredient is needed: the \emph{Three-in-a-Tree} Theorem by Chudnovsky and Seymour~\cite{ChudnovskyS10}.
Let $G$ be a graph and let $x,y,z \in V(G)$ be three distinguished vertices. 
The Three-in-a-Tree Theorem provides a dichotomy:
either we can find an induced tree in $G$ that contains $x$, $y$, and $z$, or we can find a suitable decomposition of $G$ that somehow ``separates'' $x$, $y$, $z$ and witnesses that no such tree exists; 
this decomposition has a similar flavor to the decomposition for claw-free graphs~\cite{ChudnovskyS05}.
By carefully combining this result with the Gy\'{a}rf\'{a}s path argument, we show that in an $H$-free graph $G$ one can either find a balanced separator containing
a small fraction of the weight of the optimum solution (in the case of a QPTAS) or of small size (in the case of a subexponential-time algorithm),
 e.g., consisting of a constant number of vertex neighborhoods,
or a decomposition coming from the Three-in-a-Tree Theorem where every part is of significantly smaller size. 
Such a decomposition allows us to recurse on every part independently and then assemble the final result from partial results using a reduction to the maximum weight matching problem.

Having obtained the statements of Theorems~\ref{thm:main} and~\ref{thm:main2} 
for $H$ being a path or a subdivided claw, we can generalize it to $H$ being a disjoint union of such graphs in a relatively simple and standard way.

\medskip

In light of Theorems~\ref{thm:main} and~\ref{thm:main2}, we conjecture the following generalization.
\begin{conjecture}\label{conj:main}
For every forest $H$ of maximum degree at most three, 
\textsc{MWIS} admits a QPTAS and a subexponential-time algorithm in the class of graphs that do not contain any subdivision of $H$ as an induced subgraph.
\end{conjecture}
Our techniques stop short of proving Conjecture~\ref{conj:main}: we are able to prove it for $H$ containing at most three vertices of degree three.
Note that this strictly generalizes the conclusions of Theorems~\ref{thm:main} and~\ref{thm:main2} for $H$ being a subdivided claw
(with a $\Oh(|V(G)|^{40/41} \log |V(G)|)$ term in the exponent of the running time bound of the subexponential-time algorithm).

Furthermore, as a side result we obtain a QPTAS and a subexponential-time algorithm for graphs excluding a long hole.
\begin{theorem}\label{thm:hole}
For every $t \geq 4$ there exists an algorithm that, given 
a graph $G$ that does not contain any cycle of length at least $t$ as an induced subgraph,
  a weight function $\wei : V(G) \to \mathbb{N}$, and an accuracy parameter $\varepsilon > 0$,
computes a $(1-\varepsilon)$-approximation to \textsc{Maximum Weight Independent Set} on $(G,\wei)$
in time exponential in a polynomial of $\log |V(G)|$ and $\varepsilon^{-1}$.
Furthermore, in the same graph class \textsc{Maximum Weight Independent Set} can be solved exactly
in time exponential in $\Oh(|V(G)|^{1/2} \log |V(G)|)$. 
\end{theorem}
The techniques of Theorem~\ref{thm:hole} allow us also to state the following graph-theoretical
corollary that generalizes an analogous result for $P_t$-free graphs~\cite{BacsoLMPTL19,GORSSS18}
and for graphs excluding any induced cycle of length at least $5$~\cite{long-hole-free-subexp}.
\begin{theorem}\label{thm:hole-tw}
For every $t \geq 4$ there exists a constant $c_t$ such that 
every graph $G$ that does not contain any cycle of length at least $t$ as an induced subgraph
has treewidth bounded by $c_t \Delta$, where $\Delta$ is the maximum degree of $G$.
\end{theorem}

\paragraph{Recent progress.}
This work, announced in 2019~\cite{arxiv,ChudnovskyPPT20},
initiated a sequence of rapid developments.

In the case of $P_t$-free graphs, Gartland and Lokshtanov showed in 2020 a quasipolynomial-time
exact algorithm~\cite{GL20} for \textsc{MWIS}; the result was later simplified
by Pilipczuk, Pilipczuk and Rz\k{a}\.{z}ewski~\cite{PPR20SOSA}.
The authors of the aforementioned two papers joined forces and generalized this result
to the class of graphs that do not contain any cycle of length at least $t$ as an induced subgraph
and to a wide family of problems that include \textsc{MWIS}
and \textsc{Feedback Vertex Set}~\cite{GLPPRz20}.

Meanwhile, in the regime of polynomial-time algorithms,
Abrishami et al~\cite{ACPRS-SODA} and Chudnovsky et al~\cite{CMPPR22}
improved the framework of~\cite{LokshtanovVV14,GrzesikKPP22}, generalizing
the algorithm for \textsc{MWIS} in $P_6$-free graphs to the same set of problems
as in~\cite{GLPPRz20} and some other related graph classes.

For the general case of $H$-free graphs where every component of $H$ is
a path or a subdivided claw,
in 2021 Abrishami et al~\cite{AbrishamiCDR22} showed a polynomial-time algorithm
with the additional assumption of bounded maximum degree.
Then, in 2022 Majewski et al~\cite{MajewskiM0OPRS22} showed a structural result
being the analog of the Gy\'{a}rf\'{a}s' path argument, providing
an arguably simpler proof of Theorems~\ref{thm:main} and~\ref{thm:main2}
with better running time bounds. 
Building on their work and the algorithms for $P_t$-free graphs~\cite{GL20,PPR20SOSA},
Gartland et al~\cite{sttt-free-qp} 
very recently announced a quasipolynomial-time exact algorithm for \textsc{MWIS}. 
Furthermore, Abrishami et al~\cite{abrishami2023max} showed
how to use the structural result of~\cite{MajewskiM0OPRS22} 
to simplify the arguments of~\cite{AbrishamiCDR22} and generalize 
the assumption of bounded maximum degree 
to the assumption of being $K_s$-free and $K_{s,s}$-free for a constant $s$.

In all aforementioned works for $P_t$-free graphs, 
the Gy\'{a}rf\'{a}s path argument is essential, while 
in all aforementioned works for graphs excluding subdivided claws,
the usage of extended strip decompositions and the three-in-a-tree theorem 
outlined in this work plays pivotal role. 
It remains at the moment highly unclear how to combine these tools with the framework
of \emph{potential maximal cliques}, being the base of the known polynomial-time algorithms
in $P_5$- and $P_6$-free graphs~\cite{LokshtanovVV14,GrzesikKPP22,ACPRS-SODA,CMPPR22}.

\paragraph{Organization.}
After brief preliminaries in Section~\ref{sec:prelims}, we present our framework in Section~\ref{sec:disperser}.
In Section~\ref{sec:heavy} we treat heavy vertices.
As a warm-up, the argument for $P_t$-free graphs is described in Section~\ref{sec:pt-free};
this section also contains proofs of 
Theorems~\ref{thm:hole} and~\ref{thm:hole-tw}.
Section~\ref{sec:spider}, the main technical part of the paper, considers the case
of $H$-free graphs where $H$ is a subdivided claw, with Theorems~\ref{thm:main} and~\ref{thm:main2} inferred in Section~\ref{ssec:thm:main}.
Finally, in Section~\ref{sec:horse} we prove Conjecture~\ref{conj:main} for $H$ being a forest with at most three
vertices of degree three.

\section{Preliminaries}\label{sec:prelims}
For an (undirected, simple) graph $G$ and a vertex $v \in V(G)$, $N(v)$ denotes the (open) neighborhood
of $v$, and $N[v] = \{v\} \cup N(v)$ is the closed neighborhood of $v$.
We extend it to sets of vertices $X \subseteq V(G)$ by $N[X] = \bigcup_{v \in X} N[v]$ and $N(X) = N[X] \setminus X$.
Whenever the graph $G$ is not clear from the context, we clarify it by putting it in the subscript.
For brevity, we sometimes identify subgraphs with their vertex set when this does not create any confusion: if $D$ is a subgraph of $G$,
then $G-D$, $N[D]$, and $N(D)$ are shorthands for $G-V(D)$, $N[V(D)]$, and $N(V(D))$, respectively.
Two disjoint sets $A,B \subseteq V(G)$ are \emph{fully adjacent} (\emph{fully anti-adjacent}) if there are all possible edges (no edges, respectively) between $A$ and $B$.
By $P_t$ we denote a path on $t$ vertices.
For a graph $G$, $\cc(G)$ is the family of connected components of $G$.

\subsection{Maximum Weight Independent Set}
Let $G$ be a graph and let $\wei : V(G) \to \mathbb{N}$ be a weight function. 
For a set $X \subseteq V(G)$ we denote $\wei(X) = \sum_{x \in X} \wei(x)$.
The \textsc{Maximum Weight Independent Set} (\textsc{MWIS}) problem asks for an independent set $I \subseteq V(G)$ maximizing $\wei(I)$.
We say that an independent set $I$ is an \emph{$\alpha$-approximation} for \textsc{MWIS} in 
$(G,\wei)$ if for every independent set $I'$ in $G$ we have $\wei(I) \geq \alpha \cdot \wei(I')$.
In this work, given $G$, $\wei$, and an accuracy parameter $\eps > 0$, we ask for an independent set $I$ that is a $(1-\eps)$-approximation.
For simplicity, we will develop an algorithm that gives only a $(1-c \cdot \eps)$-approximation for some universal constant $c$, as we can then use it with rescaled value of $\eps$.
We denote $n = |V(G)|$.

\subsection{Extended strip decomposition and the three-in-a-tree theorem}
Let $G$ be a graph. An \emph{extended strip decomposition} of $G$ consists of the following:
\begin{enumerate}
\item a simple non-empty graph $H$, 
\item a vertex set $\esd(e) \subseteq V(G)$ for every $uv = e \in E(H)$ and subsets $\esd(e,u), \esd(e,v) \subseteq \esd(e)$,
\item a vertex set $\esd(v) \subseteq V(G)$ for every $v \in V(H)$, and
\item a vertex set $\esd(T) \subseteq V(G)$ for every triangle $T$ in $H$,
\end{enumerate}
with the following properties:
\begin{enumerate}
\item the vertex sets of $\esd(e)$, $\esd(v)$, and $\esd(T)$ form a partition of $V(G)$;
\item for every $v \in V(H)$ and every two distinct edges $vu, vw \in E(H)$ incident with $v$, the set $\esd(vu,v)$ is fully adjacent to $\esd(vw,v)$ in $G$;
\item every edge $xy \in E(G)$ is either contained in one of the graphs $G[\esd(e)]$, $G[\esd(v)]$, $G[\esd(T)]$, or is one of the following types:
\begin{itemize}
\item $x \in \esd(e,v)$, $y \in \esd(e',v)$ for two distinct edges $e$, $e'$ of $H$ incident with a common vertex $v \in V(H)$;
\item $x \in \esd(v)$ and $y \in \esd(e,v)$ for some edge $e \in E(H)$ incident with a vertex $v \in V(H)$;
\item $x \in \esd(T)$ and $y \in \esd(e,v) \cap \esd(e,u)$ for some triangle $T$ in $H$ and an edge $e = uv$ of this triangle.
\end{itemize}
\end{enumerate}

The main result of~\cite{ChudnovskyS10} is the following.
\begin{theorem}[\hspace{1sp}\cite{ChudnovskyS10}]\label{thm:3-in-a-tree}
Let $G$ be a connected graph and let $Z \subseteq V(G)$ be a set of size at least two such that for every induced tree $T$
of $G$, $|V(T) \cap Z| \leq 2$.
Then there exists an extended strip decomposition $(H,\esd)$ of $G$ such that for every $z \in Z$
there exists a distinct vertex $w_z \in V(H)$ of degree one in $H$ with
$\esd(e_z, w_z) = \{z\}$ where $e_z$ is the unique edge of $H$ incident with $w_z$.
Furthermore, given $G$ and $Z$, such a decomposition can be computed in polynomial time.
\end{theorem}

Given a graph $G$ and an extended strip decomposition $(H,\esd)$ of $G$, a vertex $z$
satisfying the property expressed in Theorem~\ref{thm:3-in-a-tree} will be called {\em{peripheral}} in $(H,\esd)$.
Concretely, $z$ is peripheral in $(H,\esd)$ if there exists a vertex $w_z$ of $H$, said to be {\em{occupied}} by $z$, such that $w_z$ has degree $1$ in $H$ and satisfies $\esd(e_z,w_z)=\{z\}$, 
where $e_z$ is the unique edge incident to $w_z$ in $H$.

We will also need the notion of a \emph{trivial} extended strip decomposition. Given a graph
$G$, a \emph{trivial extended strip decomposition} $(H,\esd)$ consists of an edgeless graph $H$
that has a vertex $x_C$ for every connected component $C$ of $G$ and $\esd(x_C) = C$.

\section{Disperser yields a QPTAS}\label{sec:disperser}
\newcommand{\triangles}{\mathcal{T}}
\newcommand{\Csafe}{\gamma}
\newcommand{\Cshrink}{\delta}
\newcommand{\algo}{\mathcal{A}}
\newcommand{\DisSize}{\mathbf{S}}
\newcommand{\DisTime}{\mathbf{T}}
\newcommand{\DisFamily}{\mathcal{D}}
\newcommand{\Ibound}{\mathbf{m}}

\newcommand{\Ehit}{\xi}

Let $G$ be a graph and let $(H,\esd)$ be an extended strip decomposition of $G$.
For an edge $e \in E(H)$, let $\triangles(e)$ be the set of all triangles of $H$ that contain $e$.
We define a number of \emph{atoms} as follows. 
For every edge $e = uv \in E(H)$, we define the following atoms:
\begin{align*}
A_e^\bot &= \esd(e) \setminus \left(\esd(e,u) \cup \esd(e, v)\right), &
A_e^u &= \left(\esd(u) \cup \esd(e)\right) \setminus \esd(e, v),\\
A_e^v &= \left(\esd(v) \cup \esd(e)\right) \setminus \esd(e, u),  &
A_e^{uv} &= \esd(u) \cup \esd(v) \cup \esd(e) \cup \bigcup_{T \in \triangles(e)} \esd(T).
\end{align*}
Furthermore, we define an atom $A_v = \esd(v)$ for every $v \in V(H)$ and an atom $A_T = \esd(T)$
for every triangle $T$ in $H$.
A \emph{trivial atom} is an atom $A_v = \esd(v)$ for an isolated vertex $v$ of $H$ with
$A_v$ being a singleton containing an isolated vertex of $G$.

Let $\wei \colon V(G) \to \mathbb{N}$ be a weight function and let $\Csafe, \Cshrink > 0$ be reals.
Let $X \subseteq V(G)$ and let $(H,\esd)$ be an extended strip decomposition of $G-X$.
We say that $(X, (H,\esd))$ is
\begin{itemize}
\item \emph{$\Cshrink$-shrinking} if for every nontrivial atom $A$ of $(H,\esd)$ we have $\wei(A) \leq (1-\Cshrink) \wei(V(G))$;
\item \emph{$\Csafe$-safe} if $\wei(X) \leq \Csafe \cdot \wei(V(G))$ and, furthermore,
  for every nontrivial atom $A$ of $(H,\esd)$ it holds that $\wei(X) \leq \Csafe \cdot \wei(V(G) \setminus A)$;
\item \emph{$(\Csafe,\Cshrink)$-good} if it is both  $\Cshrink$-shrinking and $\Csafe$-safe.
\end{itemize}
For a set $I \subseteq V(G)$, a weight function $\wei_I$ is defined as $\wei_I(v) = \wei(v)$ for every $v \in I$ and $\wei_I(v) = 0$ for every $v \in V(G) \setminus I$.

For approximation schemes, we need the following notion.

\begin{definition}
For a graph $G$ and a weight function $\wei \colon V(G) \to \mathbb{N}$,
a \emph{$(\Csafe,\Cshrink)$-disperser} is a family $\Dd$ such that:
\begin{itemize}
\item every member of $\Dd$ is a pair of the form $(X,(H,\esd))$, where $(H,\esd)$ is an extended strip decomposition of $G-X$; and 
\item for every independent set $I$ in $G$ with $\wei(I) > 0$ there exists $(X,(H,\esd)) \in \Dd$ that is $(\Csafe,\Cshrink)$-good for $G$ and $\wei_I$.
\end{itemize}
\end{definition}

If one is interested in subexponential-time algorithms, it suffices to 
consider the following simpler notion that considers only uniform weights.
\begin{definition}
For a constant $\Ehit \in (0,1)$ and a graph $G$, a \emph{$\Ehit$-uniform disperser} is a pair $(X,(H,\esd))$, where $X \subseteq V(G)$ and $(H,\esd)$ is an extended strip
decomposition of $G-X$ such that 
$$|X| \leq |V(G)|^{-\Ehit}\cdot |V(G)\setminus A|\qquad\textrm{and}\qquad |A| \leq |V(G)|-|V(G)|^\Ehit\qquad \textrm{for every atom }A\textrm{ of }(H,\esd).$$
\end{definition}

\subsection{Intuition}
The main result of this section is that an algorithm producing dispersers with good parameters
yields a QPTAS and, similarly, an algorithm producing uniform dispersers with good parameters
yields an exact subexponential-time algorithm. Let us now give some intuition.

Let $G$ be a graph and let $(H,\esd)$ be an extended strip decomposition of $G$. 
Let $A_1$ and $A_2$ be two atoms of $(H,\esd)$. 
We say that \emph{$A_1$ and $A_2$ are conflicting} if they are potentially not disjoint; that is, for every $e = uv \in E(H)$
\begin{enumerate}[(i)]
\item $A_e^\bot$, $A_e^u$, $A_e^v$, and $A_e^{uv}$ are pairwise in conflict;\label{p:cf:1}
\item both $A_e^u$ and $A_e^{uv}$ conflict with $A_u$ and both $A_e^v$ and $A_e^{uv}$ conflict with $A_v$;\label{p:cf:2}
\item $A_e^{uv}$ and $A_e^u$ conflicts with $A_{e'}^{uv'}$ and $A_{e'}^u$ for every edge $e' = uv' \in E(H)$ incident with $u$, and similarly for the $v$ endpoint; and\label{p:cf:3}
\item $A_e^{uv}$ and $A_T$ are in conflict for every $T \in \triangles(e)$.\label{p:cf:4}
\end{enumerate}
Observe that if $A_1$ and $A_2$ are not conflicting then not only $A_1 \cap A_2 = \emptyset$ but also $E(A_1,A_2) = \emptyset$.
Informally, two atoms $A_1$ and $A_2$ are not conflicting if and only if the definition of the extended strip decomposition
ensures that they are disjoint and there is no edge of $G$ between $A_1$ and $A_2$.
A family $\mathcal{A}$ of atoms of $(H,\esd)$ is \emph{independent} if every two distinct elements of $\mathcal{A}$ are not conflicting. 

For an independent set $I$ in $G$, we define the following family $\mathcal{A}_I$ of atoms of $(H,\esd)$:
\begin{itemize}
\item $A_e^{uv}$ for every $e = uv \in E(H)$ with $I \cap \esd(e,u) \neq \emptyset$ and $I \cap \esd(e,v) \neq \emptyset$,
\item $A_e^u$ for every $e = uv \in E(H)$ with $I \cap \esd(e,u) \neq \emptyset$ but
$I \cap \esd(e,v) = \emptyset$,
\item $A_e^v$ for every $e = uv \in E(H)$ with $I \cap \esd(e,v) \neq \emptyset$ but
$I \cap \esd(e,u) = \emptyset$,
\item $A_e^\bot$ for every $e = uv \in E(H)$ with $I \cap (\esd(e,u) \cup \esd(e, v)) = \emptyset$,
\item $A_v$ for every $v \in V(H)$ such that for every $e$ incident with $v$ we have
$I \cap \esd(e, v) = \emptyset$,
\item $A_T$ for every triangle $T$ in $H$ such that for all edges $e = uv$ of $T$ we have
$I \cap \esd(e, u) = \emptyset$ or $I \cap \esd(e, v) = \emptyset$.
\end{itemize}
Observe that for every $v \in V(H)$, $I$ may intersect at most one set $\esd(e,v)$ for $e$ incident with $v$. 
From this, a direct check verifies the following crucial observation:
\begin{claim}\label{cl:dis:partI}
For every independent set $I$ in $G$, the family $\mathcal{A}_I$ is independent and $I \subseteq \bigcup \mathcal{A}_I$.
\end{claim}
\begin{proof}
We consider the four cases of how the atoms can be conflicting one-by-one.
For Case~\ref{p:cf:1}, observe that for every $e = uv \in E(H)$, the 
conditions for $A_e^\bot$, $A_e^u$, $A_e^v$, $A_e^{uv}$ are mutually exclusive and exactly one of these
atoms is in $\mathcal{A}_I$.
For Case~\ref{p:cf:2}, by definition $A_v \in \mathcal{A}_I$ only if 
$A_e^v, A_e^{uv} \notin \mathcal{A}_I$ for every edge $e = uv$ incident with $v$.

Case~\ref{p:cf:3} is the most interesting: the definition of the extended strip decomposition 
ensures that $\esd(e,v)$ and $\esd(e',v)$ are fully adjacent for two different edges $e,e'$ incident with $v$,
        and thus for every $v \in V(H)$ the independent set 
 $I$ can contain a vertex of at most one set $\esd(e,v)$ over all edges $e$ incident with $v$.
 Consequently, $\mathcal{A}_I$ contains at most one set $A_e^{uv}$ or $A_e^v$ over all edges $e = uv$ incident with $v$.

Finally, for Case~\ref{p:cf:4}, $A_T$ is conflicting only with atoms $A_e^{uv}$ for edges $e = uv$ of $T$,
 but the condition for including $A_T$ into $\mathcal{A}_I$ is a negation of the condition for excluding
any $A_e^{uv}$ for edges $e=uv$ of $T$.
\cqed\end{proof}

In the other direction, if we are given an independent set $I(A) \subseteq A$ for every atom $A \in \mathcal{A}$
of an independent family $\mathcal{A}$ of atoms, then $\bigcup_{A \in \mathcal{A}} I(A)$ is an independent set in $G$.

Thus, one can reduce finding a (good approximation of) maximum-weight independent set in $G$ to finding 
such a (good approximation of) independent set in subgraphs $G[A]$ for atoms $A \in \mathcal{A}_I$, where $I$ is the sought maximum-weight independent set. 
In the definition of a disperser, if one recurses in the above sense on $G-X$ and $(H,\esd)$ for every $(X,(H,\esd))$ in the disperser,
the notion of $\Cshrink$-shrinking ensures that such recursion is of small depth, while the notion
of $\Csafe$-safety ensures that by sacrificing the set $X$ we lose only a small fraction of the optimum at every recursion step. 
In uniform dispersers, the bound on the size of $X$ allows us to branch exhaustively on $X$ in the recursion step; this cost is amortized by the decrease in the size
of graphs considered in the branches.

However, there is one major obstacle to the above outline: we do not know the family $\mathcal{A}_I$. Instead, we can recurse on every atom of $(H,\esd)$.

Then, we need an observation that assembling results from the recursion in the best possible way reduces to a maximum-weight matching problem in an auxiliary graph, in a similar fashion that finding maximum-weight independent set in line graphs corresponds
to finding maximum-weight matching in the preimage graph.

\subsection{Formal statements}
The following definition encompasses the idea that a graph class admits efficiently computable dispersers.

\begin{definition}\label{def:disperser}
Let $\Csafe \in (0,1/2)$ be a real, $\Cshrink \colon \mathbb{N} \to (0,1/2)$ be a nonincreasing function,
and $\DisSize, \DisTime \colon \mathbb{N} \to \mathbb{N}$ be nondecreasing functions.
A hereditary graph class $\Cc$ is called {\em{$(\Csafe,\Cshrink,\DisSize,\DisTime)$-dispersible}} if there exists an algorithm
that, given an $n$-vertex graph $G\in \Cc$ and a weight function $\wei \colon V(G) \to \mathbb{N}$, 
runs in time $\DisTime(n)$ and computes a $(\Csafe, \Cshrink(n))$-disperser for $G$ and $\wei$ of size at most~$\DisSize(n)$.
\end{definition}

The main theorem concerning approximation schemes %
  is the following.

\begin{theorem}\label{thm:disperser}
Let $\Cc$ be a hereditary graph class with the following property: For every $\Csafe\in (0,1/2)$ there exist functions $\Cshrink,\DisSize,\DisTime$
where
$$(\Cshrink(n))^{-1}\in \poly(\log n,\Csafe^{-1})\qquad\textrm{and}\qquad \DisSize(n),\DisTime(n)\in 2^{\poly(\log n,\Csafe^{-1})}$$
and $\Cshrink(n)$ is computable in polynomial time given $\Csafe$ and $n$,
such that $\Cc$ is $(\Csafe,\Cshrink,\DisSize,\DisTime)$-dispersible.
Then \textsc{MWIS} restricted to graphs from $\Cc$ admits a QPTAS.
\end{theorem}

From now on, hereditary classes $\Cc$ satisfying the assumptions of Theorem~\ref{thm:disperser} will be called {\em{QP-dispersible}}.
Thus, Theorem~\ref{thm:disperser} states that \textsc{MWIS} admits a QPTAS on every QP-dispersible class,
while in the next sections we will prove that several classes are indeed QP-dispersible.

The above definitions are suited for all our results, but in some simpler cases we will construct dispersers that have a simpler form.
More precisely, a disperser $\Dd$ is {\em{strong}} if for each $(X,(H,\esd))\in \Dd$, $(H,\esd)$ is the trivial extended strip decomposition of $G-X$. 
Recall that this means that $(H,\esd)$ simply decomposes $G-X$ into connected components:
$H$ is an edgeless graph with vertices mapped bijectively to connected components of $G-X$; then the atoms of $(H,\esd)$ are exactly the connected components of $G-X$.
As for strong dispersers the decomposition $(H,\esd)$ is uniquely determined by $X$, we will somewhat abuse notation and regard strong dispersers as simply families of sets $X$, instead of pairs
of the form $(X,(H,\esd))$. 
Intuitively, a strong disperser for $G$ is simply a family of subsets of vertices such that for every possible weight function $\wei$, some member of the family is a balanced separator for $\wei$ 
that has a small weight by itself.
The notions of QP-dispersibility lifts to {\em{strong QP-dispersibility}} by considering strong dispersers instead of regular ones.

Similarly, uniform dispersers imply subexponential-time algorithms.
\begin{theorem}\label{thm:uniform-disperser}
Let $\Cc$ be a hereditary graph class with the following property:
there exist constants $n_0 > 0$, $\tau > 0$, and $\Ehit \in (0,1)$ and an algorithm that,
given a connected graph $G \in \Cc$ with $n \geq n_0$ vertices
and such that $|N_G[v]| \leq \tau n^{\Ehit}$ for every $v \in V(G)$, outputs in polynomial time a 
$\Ehit$-uniform disperser for $G$.
Then, \textsc{MWIS} restricted to graphs from $\Cc$
admits an algorithm with time complexity $2^{\Oh(n^{1-\Ehit} \log n)}$.
\end{theorem}
In Theorem~\ref{thm:uniform-disperser}, the constant hidden in the big-$\Oh$ notation may depend on $n_0$, $\tau$, and $\Ehit$. 
A hereditary graph class $\Cc$ satisfying the assumptions of Theorem~\ref{thm:uniform-disperser} for $n_0$, $\tau$, and $\Ehit$
is called \emph{$\Ehit$-uniformly dispersible}. 

The rest of this section is devoted to the proofs of Theorems~\ref{thm:disperser}
and~\ref{thm:uniform-disperser}.

\subsection{Using maximum-weight matching}\label{ss:matching}
Assume that a graph $G$ is equipped with a weight function $\wei$ and
an extended strip decomposition $(H,\esd)$. 
Furthermore, for every atom $A$ of $(H,\esd)$ we are given an independent set $I(A) \subseteq A$. 

Construct a graph $H'$ as follows: start with the graph $H$ and then, for every edge $e = uv$
of $H$, add a new vertex $x_e$ and edges $x_eu$ and $x_ev$. 
Furthermore, define weight function $\wei'$ on $E(H')$ as follows:
\begin{align*}
\wei'(x_e u) &= \wei(I(A_e^u)) - \wei(I(A_u)) - \wei(I(A_e^\bot)),\\
\wei'(x_e v) &= \wei(I(A_e^v)) - \wei(I(A_v)) - \wei(I(A_e^\bot)),\\
\wei'(e) &= \wei(I(A_e^{uv})) - \wei(I(A_u)) - \wei(I(A_v)) - \wei(I(A_e^\bot)) - \sum_{T \in \triangles(e)} \wei(I(A_T)).
\end{align*}
We claim that the problem of finding maximum-weight matching in $(H',\wei')$ is closely related
to the problem of finding \textsc{MWIS} in $(G,\wei)$.
Let
$$a = \sum_{v \in V(H)} \wei(I(A_v)) + \sum_{e \in E(H)} \wei(I(A_e^\bot)) + \sum_{T \in \triangles(H)} \wei(I(A_T)).$$
For a family $\mathcal{A}$ of atoms of $(H,\esd)$, we define $M(\mathcal{A}) \subseteq E(H')$ as follows.
For every $e = uv \in E(H)$, we insert into $M(\mathcal{A})$:
\begin{itemize}
\item the edge $e$ if $A_e^{uv} \in \mathcal{A}$,
\item the edge $x_e u$ if $A_e^u \in \mathcal{A}$, and
\item the edge $x_e v$ if $A_e^v \in \mathcal{A}$. 
\end{itemize}
A direct check shows the following.
\begin{claim}\label{cl:dis:atoms2matching}
If $\mathcal{A}$ is an independent family of atoms of $(H,\esd)$, then $M(\mathcal{A})$ is a matching in $H'$.
Furthermore,
 \begin{equation}\label{eq:dis:atoms2matching}
 \wei'(M(\mathcal{A})) \geq -a + \sum_{A \in \mathcal{A}} \wei(I(A)).
 \end{equation}
\end{claim}
\begin{proof}
First we verify that $M(\mathcal{A})$ is a matching in $H'$. 
From the definition of independent set of atoms we infer that for every $e = uv \in E(H)$ at most 
one of the edges $e$, $x_eu$, or $x_ev$ belongs to $M(\mathcal{A})$. Furthermore, if
$x_eu$ or $e$ belongs to $M(\mathcal{A})$, we have $A_e^{uv}$ or $A_e^u$ belonging to $\mathcal{A}$, 
from which we infer that neither $A_u$ nor $A_{e'}^{uv'}$ nor $A_{e'}^u$ belongs to $\mathcal{A}$ for any other $e' = uv' \in E(H)$
incident with $u$ in $H$.
In particular, neither $e'$ nor $x_{e'}u$ belongs to $M(\mathcal{A})$.
Also, if $A_e^{uv} \in \mathcal{A}$ and $T \in \triangles(e)$, then $A_T \notin \mathcal{A}$ 
and $A_{e'}^{u'v'} \notin \mathcal{A}$ for every other edge $e' = u'v'$ of $T$.

For the weight bound, we consider their contribution to the left and right hand side of~\eqref{eq:dis:atoms2matching} one-by-one.
\begin{itemize}
\item for every atom $A$ of the form $A_e^{uv}$, $A_e^u$, or $A_e^v$,
 \begin{itemize} 
\item if $A \in \mathcal{A}$, then  the term $\wei(I(A))$ 
appears once on the left hand side and once on the right hand side,
\item if $A \notin \mathcal{A}$, then the term $\wei(I(A))$ does 
not appear at all in~\eqref{eq:dis:atoms2matching};
\end{itemize}
\item for every $e =uv \in E(H)$, 
  \begin{itemize}
  \item if $A_e^\bot \in \mathcal{A}$, then the term $\wei(I(A_e^\bot))$ does not appear
on the left hand side (as then $A_e^u,A_e^v,A_e^{uv} \notin \mathcal{A}$)
  and its appearances on right hand side in $a$ and $\sum_{A \in \mathcal{A}} \wei(I(A))$ cancel out,
  \item if $A_e^\bot \notin \mathcal{A}$, then the term $\wei(I(A_e^\bot))$ appears with $-1$ coefficient
  on the right hand side (in the $-a$ term), while on the left hand side it appears with $-1$ coefficient
  if $A_e^u$, $A_e^v$, or $A_e^{uv}$ belongs to $\mathcal{A}$, and does not appear at all otherwise. 
  \end{itemize}
\item for every $v \in V(H)$, 
  \begin{itemize}
  \item if $A_v \in \mathcal{A}$, 
  then the appearances of $\wei(I(A_v))$ on the right hand side cancel out,
  while this term does not appear on the left hand side (the definition of independence ensures that
      no atom $A_e^v$ nor $A_e^{uv}$ is in $\mathcal{A}$ for any edge $e = uv$ incident with $v$),
  \item if $A_v \notin \mathcal{A}$, then $\wei(I(A_v))$ appears with $-1$ coefficient on the right hand side,
  while the independence of $\mathcal{A}$ implies that for at most one edge $e = uv$ incident with $v$
  the atom $A_e^v$ or $A_e^{uv}$ belongs to $\mathcal{A}$ and, consequently, $\wei(I(A_v))$ either does
  not appear on the left hand side or appears once with $-1$ coefficient;
  \end{itemize}
\item for every triangle $T$ in $H$,
  \begin{itemize}
  \item if $A_T \in \mathcal{A}$, 
  then the appearances of $\wei(I(A_T))$ on the right hand side cancel out,
  while this term does not appear on the left hand side (the definition of independence ensures that
      no atom $A_e^{uv}$ is in $\mathcal{A}$ for any edge $e = uv$ of $T$),
  \item if $A_T \notin \mathcal{A}$, then $\wei(I(A_T))$ appears with $-1$ coefficient on the right hand side,
  while the independence of $\mathcal{A}$ implies that for at most one edge $e = uv$ of $T$
  the atom $A_e^{uv}$ belongs to $\mathcal{A}$ and, consequently, $\wei(I(A_T))$ either does
  not appear on the left hand side or appears once with $-1$ coefficient.
  \end{itemize}
\end{itemize}
Thus, we have shown that for every atom $A$, the coefficient in front of $\wei(I(A))$ on the left hand side
of~\eqref{eq:dis:atoms2matching} is not smaller than the coefficient on the right hand side. This finishes the proof of the claim.
\cqed\end{proof}
In the other direction, for $M \subseteq E(H')$ define a family $\mathcal{A}(M)$ of atoms of $G$ as follows.
\begin{itemize}
\item For every edge $e = uv \in E(H) \cap M$, insert $A_e^{uv}$ into $\mathcal{A}(M)$.
\item For every edge $x_e u \in M \setminus E(H)$, insert $A_e^u$ into $\mathcal{A}(M)$.
\item For every edge $e = uv \in E(H)$ such that neither $e$, $x_e u$, nor $x_e v$ is in $H$, insert $A_e^\bot$ into $\mathcal{A}(M)$.
\item For every vertex $v \in V(H)$ such that none of the edges of $M$ is incident with $v$, insert $A_v$ into $\mathcal{A}(M)$.
\item For every triangle $T$ in $H$ such that none of the edges of $H$ is in $M$, insert $A_T$ into $\mathcal{A}(M)$.
\end{itemize}
Again, a direct check shows the following.
\begin{claim}\label{cl:dis:matching2atoms}
If $M$ is a matching in $H'$, then $\mathcal{A}(M)$ is an independent family of atoms of $(H,\esd)$.
Furthermore,
\begin{equation}\label{eq:dis:matching2atoms}
\sum_{A \in \mathcal{A}(M)} \wei(I(A)) = a + \wei'(M).
\end{equation}
\end{claim}
\begin{proof}
To show that $\mathcal{A}(M)$ is independent, we consider the cases how two atoms can be conflicting one-by-one.
For Case~\ref{p:cf:1}, since at most one edge $e$, $x_eu$, $x_ev$ for $e=uv \in E(H)$ belongs to $M$,
we have that exactly one of the atoms $A_e^\bot$, $A_e^u$, $A_e^v$, $A_e^{uv}$ belongs to $\mathcal{A}(M)$.
For Case~\ref{p:cf:2}, we insert $A_v$ into $\mathcal{A}(M)$ only if none of the edges of $M$ is incident with $v$,
    which in particular implies that neither $A_e^v$ nor $A_e^{uv}$ is in $\mathcal{A}(M)$ for any edge $e = uv \in E(H)$
    incident with $v$.
For Case~\ref{p:cf:3}, since $M$ is a matching, for every $u \in V(H)$ and two distinct edges $e=uv$ and $e'=uv'$
incident with $u$ in $H$, at most one of the edges $e$, $e'$, $x_eu$, and $x_{e'}u$ belong to $M$, and thus
at most one of the atoms $A_e^u$, $A_e^{uv}$, $A_{e'}^u$, and $A_{e'}^{uv'}$ belong to $\mathcal{A}(M)$.
Finally, for Case~\ref{p:cf:4}, if $A_T \in \mathcal{A}(M)$, then none of the edges of $T$ are in $M$
and thus no atom $A_e^{uv}$ for $e=uv$ of $T$ is in $\mathcal{A}(M)$.

For the weight bound, we consider atoms and their contribution to~\eqref{eq:dis:matching2atoms} one-by-one.
\begin{itemize}
\item for every atom $A_e^{uv}$ for $e = uv \in E(H)$,
 \begin{itemize} 
 \item if $e \in M$, then the term $\wei(I(A_e^{uv}))$ appears once 
 on the left hand side of~\eqref{eq:dis:matching2atoms} (as $A_e^{uv} \in \mathcal{A}(M)$)
 and once on the right hand side (as a part of $\wei'(e)$),
 \item if $e \notin M$, then the term $\wei(I(A_e^{uv}))$ does not appear at all in~\eqref{eq:dis:matching2atoms};
\end{itemize}
\item for every atom $A_e^{u}$ for $e = uv \in E(H)$,
 \begin{itemize} 
 \item if $x_eu \in M$, then the term $\wei(I(A_e^{u}))$ appears once 
 on the left hand side of~\eqref{eq:dis:matching2atoms} (as $A_e^{u} \in \mathcal{A}(M)$)
 and once on the right hand side (as a part of $\wei'(x_e u)$),
 \item if $e \notin M$, then the term $\wei(I(A_e^{u}))$ does not appear at all in~\eqref{eq:dis:matching2atoms};
\end{itemize}
\item for every atom $A_e^\bot$ for $e =uv \in E(H)$, 
  \begin{itemize}
  \item if neither of the edges $x_eu$, $x_ev$, or $e$ belongs to $M$, then
   $A_e^\bot \in \mathcal{A}(M)$ and the term $\wei(I(A_e^\bot))$ appears once on the left hand side of~\eqref{eq:dis:matching2atoms},
   while appearing once on the right hand side (once in $a$ and not appearing in $\wei'(M)$),
  \item if one of the edges $x_eu$, $x_ev$, or $e$ belongs to $M$, then the corresponding
  atom $A$ being $A_e^u$, $A_e^v$, or $A_e^{uv}$, respectively, belongs to $\mathcal{A}(M)$,
  and the term $\wei(I(A_e^\bot))$ does not appear on the left hand side
  while its appearances on the right hand side cancel out with the coefficient $+1$ in the term $a$
  and coefficient $-1$ in the term $\wei'(x_eu)$, $\wei'(x_ev)$, or $\wei(e)$, respectively;
  \end{itemize}
\item for every atom $A_v$ for $v \in V(H)$, 
  \begin{itemize}
  \item if there is an edge of $M$ incident with $v$, say $x_ev$ or $e$ for some $e = uv \in E(H)$,
  then $\wei(I(A_v))$ does not appear on the left hand side of~\eqref{eq:dis:matching2atoms},
  while the appearances of $\wei(I(A_v))$ on the right hand side cancel out
  with the coefficient $+1$ in the term $a$ and coefficient $-1$ in the term $\wei'(x_ev)$ or $\wei'(e)$, respectively,
  \item if there is no edge of $M$ incident with $v$, then
  $A_v \in \mathcal{A}(M)$ and term $\wei(I(A_v))$ appears once on the left hand side,
  while it appears once in $a$ on the right hand side and does not appear in $\wei'(M)$;
  \end{itemize}
\item for every atom $A_T$ for a triangle $T$ in $H$,
  \begin{itemize}
  \item if there is an edge $e$ of $T$ in $M$, then 
  $\wei(I(A_T))$ does not appear on the left hand side of~\eqref{eq:dis:matching2atoms},
  while the appearances of $\wei(I(A_T))$ on the right hand side cancel out
  with the coefficient $+1$ in the term $a$ and coefficient $-1$ in the term $\wei'(e)$,
  \item if no edges of $T$ belong to $M$, 
  then $A_T \in \mathcal{A}(M)$ and the term $\wei(I(A_T))$ appears once on the left hand side, 
  while on the right hand side it appears once in $a$ and does not appear in $\wei'(M)$.
  \end{itemize}
\end{itemize}
Thus, we have shown that for every atom $A$, the coefficient in front of $\wei(I(A))$ on the left hand side
of~\eqref{eq:dis:matching2atoms} is equal to the one on the right hand side. This finishes the proof of the claim.
\cqed\end{proof}

\subsection{Proof of Theorem~\ref{thm:disperser}}

The algorithm of Theorem~\ref{thm:disperser} is a standard recursive divide-and-conquer procedure.
Let $G \in \Cc$ be an input graph and $\wei$ be a weight function.
Fix an accuracy constant $\eps > 0$; w.l.o.g. assume that $1/\eps$ is an integer.

Since we are aiming at an approximation algorithm, we can limit the stretch of the weights value. 
The problem is trivial if $\wei(v) = 0$ for every $v \in V(G)$, so assume otherwise.
First, rescale the weight function $\wei$ such that $\max_{v \in V(G)} \wei(v) = n/\eps$ (allowing rational values of weights).
Second, round each weight down to the nearest integer value; since there exists an independent set in $G$ of weight at least $n/\eps$ (take the vertex with maximum weight), this decreases
the weight of the maximum-weight independent set by a factor of at least $(1-n \cdot \eps/n) = (1-\eps)$. 
Third, discard all vertices of $G$ of weight $0$. Consequently, we can assume that on input
the values of $\wei$ are integers within range~$[1,n/\eps]$.

Initially, we set up an upper bound $\Ibound := n^2/\eps$ on the weight of any independent set
in $G$ and $\wei$ and fix $\Csafe := \eps / (1+\log(n^2 / \eps))$.
In a recursive call, we are given an induced subgraph $G'$ of $G$ with the goal to output an independent set $I'$ in $G'$ (that, as we will prove, will be a good approximation).
We also pass to a recursive call an upper bound $\Ibound'$ on the weight of the sought independent set.

In the base of the recursion, if $G'$ is edgeless, then we return $I' = V(G')$.
Also, if $\Ibound' < 1$, then we return $I' = \emptyset$.
In the recursive step, we use the fact that $\Cc$ is QP-dispersible: for the parameter $\Csafe$ fixed above,
   there are functions $\Cshrink,\DisSize,\DisTime$ with
$$(\Cshrink(x))^{-1}\in \poly(\log x,\varepsilon^{-1})\qquad\textrm{and}\qquad \DisSize(x),\DisTime(x)\in 2^{\poly(\log x,\varepsilon^{-1})}$$
such that $\Cc$ is $(\Csafe,\Cshrink,\DisSize,\DisTime)$-dispersible.
We compute a $(\Csafe, \Cshrink(|V(G')|))$-disperser $\DisFamily$ for $(G',\wei|_{V(G')})$.

For every $(X,(H,\esd)) \in \DisFamily$, we recurse on every atom $A$ of $(H,\esd)$, passing an upper bound of $\Ibound' \cdot (1-\Cshrink(|V(G')|))$, obtaining an independent set $I(A)$.
As explained in Section~\ref{ss:matching}, we construct the graph $H'$ from $H$ and weight function $\wei'$ on $E(H')$ using
independent sets $I(A)$. 
We find a matching $M$ in $H'$ with maximum weight with respect to $\wei'$. 
We define $I_{(X,(H,\esd))} = \bigcup_{A \in \mathcal{A}(M)} I(A)$.
Finally, we return the produced independent set $I_{(X,(H,\esd))}$ of maximum weight among all elements $(X,(H,\esd)) \in \DisFamily$.

\paragraph{Running time bound.}
Since $\Cshrink$ is a nonincreasing function, $\Ibound'$ drops below $1$ at recursion depth $\Oh((\Cshrink(n))^{-1} \log (n^2/\eps))$. 
Since the sets $\esd(e)$, $\esd(v)$, and $\esd(T)$ are pairwise disjoint, there are at most $5n$ nonempty atoms in every $(H,\esd)$ for $(X,(H,\esd)) \in \DisFamily$. 
Consequently, the recursion tree has size bounded by
$$\left( \DisSize(n) \cdot 5n \right)^{\Oh((\Cshrink(n))^{-1} \log(n^2/\eps))}.$$
At every step, we spend $\DisTime(n)$ time to compute $\DisFamily$, polynomial in $n$ time to compute $\Cshrink(n)$, and $\DisSize(n) \cdot n^{\Oh(1)}$ time to handle simple manipulations
of $\DisFamily$ and the find maximum-weight matchings in $H'$. 
Hence, the algorithm runs in time bounded by an exponential function
of a polynomial in $\log n$ and $\eps^{-1}$.

\paragraph{Approximation guarantee.}
Let $I_0$ be an independent set in $G$ of maximum weight. 
We mark some recursion calls. Initially we mark the initial root call for $G$.
Consider a marked step of the recursion with subgraph $G'$.
Let $(X_0,(H_0,\esd_0))$ be an element of the computed disperser $\DisFamily$ that is $(\Csafe,\Cshrink(|V(G')|))$-good
for $G'$ and $\wei_{I_0 \cap V(G')}$; we henceforth call $(X_0,(H_0,\esd_0))$ the \emph{correct element} of the considered recursive call. 
Consider the family of atoms $\mathcal{A}_{I_0 \cap V(G'-X_0)}$ for the extended
strip decomposition $(H_0,\esd_0)$ of $G'-X_0$ and the independent set $I_0 \cap V(G'-X_0)$.
Claim~\ref{cl:dis:partI} ensures that $\mathcal{A}_{I_0 \cap V(G'-X_0)}$ is independent and its union contains
$I_0 \cap V(G'-X_0)$. 
We mark all recursive calls (being children of the recursive call for $G'$)
for atoms $A \in \mathcal{A}_{I_0 \cap V(G'-X_0)}$.

Due to our weight rescaling and rounding, 
initially $\wei(I_0) \leq n^2/\eps$.
By a straightforward top-to-bottom induction on the recursion tree,
using the definition of being $\Cshrink$-shrinking, we show that
at every marked recursive call, if $G'$ is the graph considered in the call 
and $\Ibound'$ is the passed upper bound,
then $\wei(I_0 \cap V(G')) \leq \Ibound'$.

In particular, whenever $\Ibound' < 1$, then $I_0 \cap V(G') = \emptyset$
as $\wei$ has range contained in $[1,n/\eps]$. 
Also, if $G'$ is edgeless, then the algorithm returns a maximum-weight independent set 
in $G'$. Consequently, at every marked leaf of the recursion with graph $G'$
the returned independent set in $G'$ is of weight at least $\wei(I_0 \cap V(G'))$.

Consider a nonleaf marked recursive call and let $G'$ be the graph considered
in this call. 
Let $(X_0,(H_0,\esd_0))$ be the correct element for this recursive call.
Furthermore, let $I(A)$ be the independent set output by every recursive call invoked by the 
considered call for atom $A$ of $(H_0,\esd_0)$. 
Claims~\ref{cl:dis:atoms2matching} and~\ref{cl:dis:matching2atoms} ensure that
the computed independent set for $(X_0,(H_0,\esd_0))$ satisfy
$$\wei(I_{(X_0,(H_0,\esd_0))}) \geq \sum_{A \in \mathcal{A}_{I_0 \cap V(G'-X_0)}} \wei(I(A)).$$
In particular, the independent set output by the considered recursive call for $G'$
is of weight at least the right hand side of the above inequality.

Let $\mathcal{X}$ be the family of all correct elements over all nonleaf marked recursive calls.
We infer that the weight of the independent set output by the root of the recursion
is at least 
$$\wei(I_0) - \sum_{(X_0,(H_0,\esd_0)) \in \mathcal{X}} \wei(I_0 \cap X_0).$$
Thus, it remains to estimate the sum of $\wei(I_0 \cap X_0)$ over all $(X_0,(H_0,\esd_0)) \in \mathcal{X}$.

Let $T$ be the subtree of the recursion tree induced by all marked calls. 
We call a nonleaf marked call $z$ \emph{strange} if every marked child of $z$ corresponds to 
a trivial atom of the correct element $(X_0,(H_0,\esd_0))$ at $z$, and \emph{normal} otherwise.

For every normal marked call $z$, denote by $f(z)$ 
the marked child call for a nontrivial atom $A$ with maximum $\wei(I_0 \cap A)$ (breaking ties arbitrarily)
   and mark the edge $z f(z)$ of $T$.
Let $F \subseteq E(T)$ be the set of marked edges. Clearly, $(V(T), F)$ is a set of upward paths
in $T$. Let $Z$ be the set of top endpoints of these paths, that is, $Z$ consists of the root
of $T$ and all recursive calls such that the edge of $T$ between the call and its parent
is not marked. For every $z \in V(T)$, let $G'_z$ be the subgraph of $G$ considered
in the call $z$.
Note that all marked leaves of $T$ that correspond to trivial atoms are in $Z$.
Let $S$ be the family of strange marked nodes.

As at every marked recursive call, the marked children of the call consider disjoint
atoms, we infer that every $v \in I_0$ is contained in at most 
$1+\log(\wei(I_0))$ graphs $G'_z$ for $z \in Z$
(in at most one leaf corresponding to a trivial atom and, for every other $z \in Z$ with 
 $v \in V(G'_z)$, the weight of the vertices of $I_0$ in $G'_z$ is at most half of the
 weight of the vertices of $I_0$ in the graph $G'$ at the parent of~$z$).

Furthermore, for every normal marked call $z$, from $\Csafe$-safeness
of the correct element $(X_0,(H_0,\esd_0))$ for $\wei_{I_0 \cap V(G'_z)}$ we
infer that 
$$\wei(X_0 \cap I_0) \leq \Csafe \cdot \left(\wei(I_0 \cap V(G'_z)) - \wei(I_0 \cap V(G'_{f(z)}))\right).$$
Summing over all nonleaf marked calls $z$ we infer that
\begin{align*}
\sum_{(X_0,(H_0,\esd_0)) \in \mathcal{X}} \wei(X_0 \cap I_0) &\leq \Csafe \cdot \sum_{z \in Z} \wei(I_0 \cap V(G'_z))
 + \Csafe \cdot \sum_{s \in S} \wei(I_0 \cap V(G_s')) \\
&\leq \frac{\eps}{1+\log(n^2/\eps)} \cdot \left( \log(\wei(I_0)) \cdot \wei(I_0) + \wei(I_0)\right) \\
&\leq \eps \cdot \wei(I_0).
\end{align*}
Consequently, the returned independent set at the root recursive call is of weight
at least $(1-\eps)\wei(I_0)$. This finishes the proof of Theorem~\ref{thm:disperser}.

\subsection{Proof of Theorem~\ref{thm:uniform-disperser}}

The algorithm of Theorem~\ref{thm:uniform-disperser} is again
a standard divide-and-conquer procedure,
simpler than in the case of Theorem~\ref{thm:disperser}.
By choosing $n_0$ appropriately we may assume that $n_0 > e^{1/\Ehit}$, i.e., $n_0^{\Ehit} > e$.

Let $G$ be the input graph with a weight function $\wei$.
If $n := |V(G)| \leq n_0$, we solve the problem by brute-force in constant time. 
If $G$ is disconnected, we recurse on every connected component. 
Otherwise, if there exists $v \in V(G)$ with $|N[v]| > \tau n^{\Ehit}$, branch
exhaustively on $v$: in one branch, delete $v$ from $G$ and recurse (consider $v$
    not included in the sought solution)
and in the second branch, delete $N[v]$ from $G$, recurse, and add $v$ to the independent set obtained from the recursive call (consider $v$ included in the sought solution).
Finally, output the one of the two obtained solutions that has larger weight.

In the remaining case we have $n = |V(G)| > n_0$ and $|N[v]| \leq \tau n^{\Ehit}$ for every $v \in V(G)$.
Invoke the assumed algorithm that outputs a pair $(X,(H,\esd))$ where $X \subseteq V(G)$
and $(H,\esd)$ is an extended strip decomposition of $G-X$ such that:
$$|X| \leq n^{-\Ehit} \left(n-|A|\right)\qquad\textrm{and}\qquad |A| \leq n-n^\Ehit\qquad \textrm{for every atom }A\textrm{ of }(H,\esd).$$

For every independent set $Y \subseteq X$, we proceed as follows.
Let $(H,\esd_Y)$ be $(H,\esd)$ {\em{restricted}} to $G-(X \cup N[Y])$.
That is, $(H,\esd_Y)$ is obtained from $(H,\esd)$ by removing all vertices of $N[Y]$ from all the sets in the image of $\esd$; 
it is straightforward to see that then $(H,\esd_Y)$ is an extended strip decomposition of $G-(X\cup N[Y])$.
We recurse on every atom $A$ of $(H,\esd_Y)$, obtaining an independent set $I_Y(A)$.
As explained in Section~\ref{ss:matching}, we construct the graph $H'$ from $H$ and weight function $\wei'$ on $E(H')$ using
independent sets $I_Y(A)$. 
We find a matching $M$ in $H'$ with maximum weight with respect to $\wei'$. 
We define $I_Y = Y \cup \bigcup_{A \in \mathcal{A}(M)} I_Y(A)$.
Finally, we return the independent set $I_Y$ that has
the maximum weight among all produced for independent sets $Y \subseteq X$.

\paragraph{Correctness.}
It is straightforward to verify that every recursive call is invoked on some induced
subgraph of $G$ and the set returned by any recursive call is an independent set.
By induction on $|V(G)|$,
we prove that an application of the algorithm to a graph $G$ returns a maximum-weight independent set in $G$.

This is obvious for the cases when we apply a brute-force search 
and when $G$ is disconnected and we recurse on the connected components of $G$.
If we branch on a vertex $v$ with $|N[v]| > \tau n^{\Ehit}$,
then the correctness is again straightforward as we consider exhaustively cases of $v$ being and not being
included in the sought solution.
Otherwise, we are in the case where we obtained a $\Ehit$-uniform disperser $(X,(H,\esd))$.
Let $I_0$ be a maximum-weight independent set in $G'$
and consider the case $Y = I_0 \cap X$. 
Then, $\mathcal{A}_{I_0 \setminus Y}$ is an independent
family of atoms of $(H,\esd_Y)$ and by Claim~\ref{cl:dis:atoms2matching},
$M(\mathcal{A}_{I_0 \setminus Y})$ is a matching in $H'$.
Furthermore, by the inductive assumption, 
for every atom $A$ of $(H,\esd_Y)$, $I_Y(A)$ is an independent set of maximum weight in $G'[A]$.
Note that we may apply the inductive assumption here due to $|A|\leq n-n^\Ehit<n$.
In particular, for every $A \in \mathcal{A}_{I_0 \setminus Y}$, 
$\wei(I_Y(A)) \geq \wei(I_0 \cap A)$. 
Therefore by Claims~\ref{cl:dis:atoms2matching} and~\ref{cl:dis:matching2atoms},
the weight of $I_Y$ is at least 
$$\wei(Y) + \sum_{A \in \mathcal{A}_{I_0 \setminus Y}} \wei(I_Y(A)) \geq \wei(Y) + \sum_{A \in \mathcal{A}_{I_0 \setminus Y}} \wei(I_0 \cap A) = \wei(I_0).$$
This concludes the inductive
  proof that the algorithm returns a maximum-weight independent set in~$G$.

\paragraph{Running time bound.}
We prove by induction on $n$ that when the algorithm is applied on an $n$-vertex graph $G$, the number of leaves of the recursion tree
is bounded by $e^{C n^{1-\Ehit} (1+\ln n)}$ for some constant $C$ depending on $\Ehit$ and $n_0$.
Since the time spent at internal computation in each recursive call is polynomial in $n$, the claimed running time bound will follow.

The claim is straightforward for the leaves of the recursion and for
non-leaf recursive calls when $G$ is disconnected.
In a non-leaf recursive call, if the algorithm branches on a vertex
$v \in V(G)$ with $|N[v]| > \tau n^{\Ehit}$, in one child recursive call
the number of vertices drops by $1$, in the second drops by at least $\tau n^{\Ehit}$.
Then the inductive step follows by standard calculations, which we omit here.%
\footnote{The crucial calculation is $n^{1-\Ehit} - (n-\tau n^\Ehit)^{1-\Ehit} \sim (1-\Ehit)\tau$.
  Alternatively, one can observe that in the recursion tree on every root-to-leaf path there are at most $\Oh(n^{1-\Ehit} \ln n/\tau)$ edges
  that correspond to branching on a vertex $v$ with $|N[v]| > \tau n^{\Ehit}$ and deleting the whole $N[v]$.}

In the remaining case, we have obtained a $\Ehit$-uniform disperser $(X,(H,\esd))$. Let $k$ be the number of vertices in the largest
atom of $(H,\esd)$. 
By the properties of $(X,(H,\esd))$, we have 
\begin{equation}\label{eq:subexp1}
|X| \leq n^{-\Ehit} \cdot (n-k) \quad\mathrm{and}\quad n-k \geq  n^\Ehit.
\end{equation}
For a fixed independent set $Y \subseteq X$, the algorithm recurses on at most $5n$ atoms, each of size at most $k$, which is strictly smaller than $n$.
Hence, by the inductive hypothesis, for a sufficiently large constant $C$ we have that the total number of leaf nodes of the recursion in descendants of the considered node is bounded by
\begin{equation}\label{eq:subexp2}
2^{|X|} \cdot 5n \cdot 2^{C \cdot k^{1-\Ehit} (1+\log k)}
\end{equation}
If $k \leq e^{1/\Ehit} \leq n_0$, then all the recursive calls are leaves in the recursion tree, so by~\eqref{eq:subexp1} their number is bounded by
$$2^{|X|} \cdot 5n \leq \exp\left(\ln 2 \cdot n^{1-\Ehit} + \ln n + \ln 5\right).$$
This value can be bounded as desired by taking $C \geq \ln 2 + 1 + \ln 5$. 
Hence, we assume 
\begin{equation}\label{eq:subexp3}
k > e^{1/\Ehit}.
\end{equation}
We need the following inequality:
\begin{align}
n^{1-\Ehit}(1+\ln n)-k^{1-\Ehit}(1+ \ln k) &\geq (n-k)\cdot \min_{k \leq x_0 \leq n} \left(\frac{d}{dx} \left(x^{1-\Ehit}(1+\ln x)\right) \Big|_{x = x_0}\right) \nonumber\\
    & = (n-k) \cdot \min_{k \leq x_0 \leq n} \left((1-\Ehit)x_0^{-\Ehit}(1+\ln x_0) + x_0^{-\Ehit}\right) \nonumber\\
    & = (n-k) n^{-\Ehit} \left((1-\Ehit)(1+\ln n) + 1\right).\label{eq:subexp4}
\end{align}
Here, in the last equality we have used~\eqref{eq:subexp3}, as $x \mapsto x^{-\Ehit}$ is decreasing for $x > 0$ and $x \mapsto x^{-\Ehit} \ln x$ is decreasing for $x \geq e^{1/\Ehit}$. 

By applying $n-k \geq n^\Ehit$,  from~\eqref{eq:subexp4} we obtain that:
\begin{equation}\label{eq:subexp5}
n^{1-\Ehit}(1+\ln n)-k^{1-\Ehit}(1+ \ln k) \geq (1-\Ehit)(n-k) n^{-\Ehit} + (1-\Ehit)(1+\ln n).
\end{equation}

With~\eqref{eq:subexp5} in hand, we are now ready to given an upper bound on~\eqref{eq:subexp2}:
\begin{align*}
&2^{|X|} \cdot 5n \cdot \exp\left(C \cdot k^{1-\Ehit} (1+\ln k)\right) \\
&\quad\leq \exp\left(\ln 2 \cdot n^{-\Ehit} \cdot (n-k) + \ln n + \ln 5 + C k^{1-\Ehit} (1+\ln k)\right) \\
&\quad\leq \exp\Big(\ln 2 \cdot n^{-\Ehit} \cdot (n-k) + \ln n + \ln 5 + C n^{1-\Ehit} (1+\ln n)\\
&\quad\quad\quad\quad\quad- C(1-\Ehit)(n-k)n^{-\Ehit} - C(1-\Ehit)(1+\ln n)\Big)\\
&\quad\leq \exp\left(C n^{1-\Ehit} (1+\ln n)\right).
\end{align*}
In the last inequality we have used~\eqref{eq:subexp5} and $C \geq \frac{\ln 5}{1-\Ehit}$.

This finishes the proof of the time complexity and of Theorem~\ref{thm:uniform-disperser}.

\section{Heavy vertices and strong dispersers}\label{sec:heavy}
Let $G$ be a graph, $\wei \colon V(G) \to \mathbb{N}$ be a weight function, and $I \subseteq V(G)$
be an independent set. For a real $\beta \in [0,1]$, a vertex $w \in V(G)$ is
\emph{$\beta$-heavy} (with respect to $I$) if $\wei(N[w] \cap I) \geq \beta \cdot \wei(I)$.
A simple coupon-collector argument shows the following.

\begin{lemma}\label{lem:heavy}
Let $G$ be an $n$-vertex graph for $n \geq 2$, $\wei \colon V(G) \to \mathbb{N}$ be a weight function, $I \subseteq V(G)$
be an independent set, and $\beta \in [0,1/2]$ be a real.
Then there exists a set $J \subseteq I$ of size at most $\lceil \beta^{-1} \log n \rceil$ such that
$N[J]$ contains all $\beta$-heavy vertices with respect to $I$.
\end{lemma}
\begin{proof}
Let $Z$ be the set of $\beta$-heavy vertices.
We consider a probability distribution on $I$ where a vertex $v \in I$ is chosen with probability
$\wei(v) / \wei(I)$.
For every $z \in Z$, a vertex $v \in I$ chosen at random according to this distribution
satisfies $z \in N[v]$ with probability at least $\beta$. 
Consequently, if $J$ is a set of $\lceil \beta \log n \rceil$ vertices of $I$ each chosen independently at random according to this distribution, then 
for every $z \in Z$ the probability that $z \notin N[J]$ is less than $(1-\beta)^{\beta \log n} < 1/n$ (here we used that $\beta \leq 1/2$ and $n \geq 2$).
By the union bound, the probability that $Z \subseteq N[J]$ is positive.
\end{proof}

Next we prove a general-usage lemma that reduces the task of finding small dispersers to connected graphs where
the neighborhood of every vertex is not $\beta$-heavy with regards to some fixed maximum-weight independent set we are looking for.
This is done essentially as follows: we first guess the set $J$ of $\beta$-heavy vertices of size $\poly(\Csafe^{-1},\log n)$ using Lemma~\ref{lem:heavy}, focus on the heaviest connected component of $G-N[J]$, 
and construct a suitable disperser for this component. This idea can be used to prove the following statement. %

\newcommand{\leps}{\sigma}

\begin{lemma}\label{lem:guess-heavy}
Let $\Cc$ be a hereditary graph class.
Suppose there is a polynomial $p(\cdot)$ such that given any $\leps>0$ and $n$-vertex connected graph $G\in \Cc$ one can in polynomial time compute 
a family $\Nn$ with $|\Nn|\leq \poly(n)$ consisting of pairs of the form $(X,(H,\esd))$, where $X\subseteq V(G)$ and $(H,\esd)$ is an extended strip decomposition of $G-X$, such that the following holds:
For every weight function $\wei\colon V(G)\to \N$ satisfying $\wei(N[v])\leq p(\leps) \wei(V(G))$ for each $v\in V(G)$ there exists $(X,(H,\esd))\in \Nn$ such that
$$\wei(A)\leq (1-p(\leps))\cdot \wei(G)\quad\textrm{and}\quad\wei(X)\leq \leps\cdot \wei(G-A)\quad\textrm{for every atom }A\textrm{ of }(H,\esd).$$
Then the class $\Cc$ is QP-dispersible. Moreover, if it is always the case that all the extended strip decompositions appearing in the family $\Nn$ are trivial (i.e. corresponding to the partition into connected components), then $\Cc$ is strongly QP-dispersible.
\end{lemma}

\begin{proof} %
Suppose without loss of generality that $p(x)\geq x$ for all positive $x$.
Fix $\Csafe\in (0,1/2)$. 
Fix $G\in \Cc$ on $n$ vertices supplied with a weight function $\wei \colon V(G) \to \N$.

We present the construction of a disperser for $G$ as a nondeterministic procedure that, for a given independent set $I$ with $\wei(I) > 0$, 
produces a pair $(X,(H,\esd))$, where $X\subseteq V(G)$ and $(H,\esd)$ is an extended strip decomposition of $G-X$, that is $(\Csafe,p(\Csafe))$-good for $\wei_I$, 
i.e. we shall have $\Cshrink(n)=p(\Csafe)$.
We argue that this nondeterministic procedure has $\DisSize(n)$ possible runs that can be enumerated in time $\DisSize(n)\cdot \poly(n)$ without the knowledge of $I$, 
where the function $\DisSize(n)$ will be chosen later.
Then the constructed disperser $\Dd$ comprises of all sets $X$ constructed by all possible runs, and thus has size at most $\DisSize(n)$.
As each run has polynomial length, the running time of the construction of $\Dd$ is $\DisTime(n)\leq \DisSize(n)\cdot \poly(n)$.

Therefore, fix an independent set $I$ in $G$ with $\wei(I)>0$. Recall that $\wei_I$ is a weight function on $G$ obtained from $\wei$ by changing the weight of vertices outside of $I$ to $0$.

First, apply Lemma~\ref{lem:heavy} to $G$, $\wei_I$, $I$, and constant $\beta=p(\Csafe)/2$. 
This yields a set $J\subseteq I$ of size at most $2p(\Csafe)^{-1}\log n+1=\poly(\Csafe^{-1},\log n)$ such that $N[J]$ contains all vertices that are $p(\Csafe)/2$-heavy w.r.t. $\wei_I$.
The procedure nondeterministically guesses the set $J$; note that there are $2^{\poly(\Csafe^{-1},\log n)}$ choices for $J$.
Then 
\begin{equation}\label{eq:no-heavy}
\wei_I(N[v])\leq p(\Csafe)/2\cdot \wei_I(G)\quad\textrm{for every vertex }v\in V(G)\setminus N[J].
\end{equation}

Let $G'$ be the heaviest (w.r.t. $\wei_I$) connected component of $G-N[J]$.
Our nondeterministic procedure guesses $G'$ ($n$ options) and whether $\wei_I(G')\leq \wei_I(G)/2$ or not ($2$ options).

Suppose first that $\wei_I(G')\leq \wei_I(G)/2$.
Then observe that putting $X=N(J)$ and $(H,\esd)$ as the trivial extended strip decomposition of $G-X$, we find that $(X,(H,\esd))$ is $(0,1/2)$-good for $G$.
Indeed, in $G-X$ every vertex of $J$ is isolated, so it corresponds to a trivial
atom of $(H,\esd)$, while every other atom of $(H,\esd)$ corresponds to a connected component of $G-N[J]$ and hence it has
weight at most $\wei_I(G)/2$. On the other hand, $\wei_I(X)=0$, because $X=N(J)$ is disjoint with $I$.

Therefore, from now on we focus on the second case when
\begin{equation}\label{eq:heavy-comp}
\wei_I(G')>\wei_I(G)/2.
\end{equation}
Since $\Cc$ is hereditary, we have $G'\in \Cc$.
Hence, we may apply the assumed algorithm to $G'$ for $\leps=\Csafe$, yielding in polynomial time a family $\Nn$ of size $\poly(n)$ consisting of pairs of the form $(X',(H',\esd'))$, 
where $(H',\esd')$ is an extended strip decomposition of $G'-X'$. 
As by~\eqref{eq:no-heavy} and~\eqref{eq:heavy-comp} we have
$$\wei_I(N_{G'}[v])\leq \wei_I(N_G[v])\leq p(\Csafe)/2\cdot \wei_I(G)\leq p(\Csafe)\cdot \wei_I(G')\quad\textrm{for every }v\in V(G'),$$ 
by assumption there exists $(X',(H',\esd'))\in \Nn$ satisfying the following:
$$\wei_I(A)\leq (1-p(\Csafe))\cdot \wei_I(G')\quad\textrm{and}\quad\wei_I(X')\leq \Csafe\cdot \wei(V(G')\setminus A)\quad\textrm{for every atom }A\textrm{ of }(H',\esd').$$
By choosing among $|\Nn|=\poly(n)$ options, our nondeterministic procedure guesses $(X',(H',\esd'))$ satisfying the above.

Consider $X=X'\cup N(J)$.
Observe that since $G'$ is a connected component of $G-N(J)$, we have $\cc(G'-X')\subseteq \cc(G-X)$.
Let now $(H,\esd)$ be the extended strip decomposition of $G-X$ obtained from $(H',\esd')$ by adding every connected component $C\in \cc(G-X)\setminus \cc(G'-X')$ as a separate piece of the decomposition:
we add a new node $w_C$ that is isolated in $H$ and set $\esd(w_C)=V(C)$.

\begin{claim}
The pair $(X,(H,\esd))$ is $(\Csafe,p(\Csafe))$-good for $G$ and $\wei_I$.
\end{claim}
\begin{proof}
First, observe that since $N(J)\cap I=\emptyset$, we have
$$\wei_I(X)=\wei_I(X')\leq \Csafe\cdot \wei_I(G'-B)\leq \Csafe\cdot \wei_I(G'),$$
where $B$ is any nontrivial atom of $(H',\esd')$.

Consider any nontrivial atom $A$ of $(H,\esd)$.
Since vertices of $J$ form trivial atoms in $(H,\esd)$, we have that either $A$ is a connected component of $G-N[J]$ that is different from $G'$, or $A$ is an atom of $(H',\esd')$.

In the first case, by~\eqref{eq:heavy-comp} we infer that $\wei_I(A)<\wei_I(G)/2$.
Moreover, since $G'$ and $A$ are disjoint, we have $\wei_I(G-A)\geq \wei_I(G')$. 
The latter assertion together with $\wei_I(X)\leq \Csafe\cdot \wei_I(G')$ implies that $\wei_I(X)\leq \Csafe\cdot \wei_I(G-A)$, as required.

Consider now the second case. First, by assumption we have $\wei_I(A)<(1-p(\Csafe))\wei_I(G')\leq (1-p(\Csafe))\wei_I(G)$.
Second, again by assumption we have $\wei_I(X)=\wei_I(X')\leq \Csafe\cdot \wei_I(G'-A)\leq \Csafe\cdot \wei_I(G-A)$.

Thus, in both cases we conclude that $(X,(H,\esd))$ is $(\Csafe,p(\Csafe))$-good for $G$ and $\wei_I$.
\cqed\end{proof}

Therefore, in all cases the nondeterministic procedure produced a pair $(X,(H,\esd))$ that is $(\Csafe,p(\Csafe))$-good for $G$ and $I$.

We conclude by observing that the nondeterminism used by the procedure comes from:
\begin{itemize}
\item choosing $J$, for which there are $2^{\poly(\Csafe^{-1},\log n)}$ choices;
\item choosing $G'$ and whether $\wei_I(G')\leq \wei_I(G)/2$, for which there are at most $2n$ choices; and
\item choosing $(X',(H',\esd'))\in \Nn$, for which there are $\poly(n)$ choices.
\end{itemize}
Hence, we can set $\DisSize(n)\in 2^{\poly(\Csafe^{-1},\log n)}$ for the size of the computed strong disperser and, consequently, also the construction running time is 
$\DisTime(n)=\DisSize(n)\cdot \poly(n)=2^{\poly(\Csafe^{-1},\log n)}$.
We conclude that $\Cc$ is $(\Csafe,p(\Csafe),2^{\mathrm{poly}(\Csafe^{-1},\log n)},2^{\mathrm{poly}(\Csafe^{-1}, \log n)})$-dispersible for every $\Csafe\in (0,1/2)$, hence it is QP-dispersible.
Moreover, it can be easily seen that if the assumed algorithm only returns trivial extended strip decompositions, then also all the constructed extended strip decompositions are trivial and, consequently,
$\Cc$ is strongly QP-dispersible.
\end{proof}

\section{Dispersers in $P_t$-free graphs and graphs without a long hole}\label{sec:pt-free}
In this section we focus on the class of $P_t$-free graphs 
and graphs excluding a long hole. 

As a warm-up, to show how our framework works, we prove the following statement. 
\begin{theorem}\label{thm:Pt-free-dispersible}
For every $t\in \N$, the class of $P_t$-free graphs is strongly QP-dispersible and $\frac{1}{2}$-uniformly dispersible.
\end{theorem}

The proof of Theorem~\ref{thm:Pt-free-dispersible} relies on a classical construction used by Gy\'arf\'as~\cite{gyarfas} to prove that $P_t$-free graphs are $\chi$-bounded, 
which is usually called the {\em{Gy\'arf\'as path}}. 
In Section~\ref{ss:pt-free} we encapsulate this concept in a versatile claim,
as we will reuse it later on.

For $t\in \N$, a graph $G$ is {\em{$C_{\geq t}$-free}} if $G$ excludes every cycle $C_\ell$ for $\ell\geq t$ as an induced subgraph.
For instance, the {\em{long-hole-free}} graphs considered in~\cite{long-hole-free-subexp} are exactly $C_{\geq 5}$-free graphs.
In Section~\ref{ss:long-hole}, we prove the following strengthening of Theorem~\ref{thm:Pt-free-dispersible} that implies Theorem~\ref{thm:hole}.  

\begin{theorem}\label{thm:hole-free-dispersible}
For every $t\in \N$, the class of $C_{\geq t}$-free graphs is strongly QP-dispersible and $\frac{1}{2}$-uniformly dispersible.
\end{theorem}

The structural results obtained in Section~\ref{ss:long-hole} also directly imply Theorem~\ref{thm:hole-tw}.

\subsection{Gy\'arf\'as' path}\label{ss:pt-free}

The following lemma encapsulates a classical construction of Gy\'arf\'as~\cite{gyarfas}.
\begin{lemma}\label{lem:gyarfas-path}
Let $\alpha\in (0,1/2)$ be a real. Let $G$ be a connected graph endowed with a weight function $\wei \colon V(G) \to \N$, and let $u$ be any vertex of $G$.
Then there is an induced path $Q=(v_0,v_1,\ldots,v_k)$ in $G$ (possibly with $k=-1$ and $Q$ being empty) such that, denoting $G_0=G-v_0$ and $G_i=G-N[v_0,\ldots,v_{i-1}]$ for $i\in \{1,\ldots,k+1\}$, 
the following holds:
\begin{enumerate}[label=(P\arabic*),ref=(P\arabic*)]
\item\label{p:init} $u=v_0$ unless $k=-1$ (where we put $G_0 = G$);
\item\label{p:fin}  for every $C\in \cc(G_{k+1})$, we have $\wei(C)\leq (1-\alpha) \wei(G)$; and
\item\label{p:mid}  for every $i\in \{0,1,\ldots,k\}$, there is a connected component $D$ of $G_i$ such that $\wei(D)>(1-\alpha)\wei(G)$ and $D$ contains a neighbor of $v_i$.
\end{enumerate}
Moreover, given $G$ and $u$ one can compute in polynomial time a family $\Qq$ consisting of $\Oh(|V(G)|^2)$ induced paths in $G$, each starting at $u$, so that for every $\alpha\in (0,1/2)$ and
weight function $\wei\colon V(G) \to \N$ there exists $Q\in \Qq$ satisfying the above properties for $\alpha$ and $\wei$.
\end{lemma}
\begin{proof}
We first prove the existential statement and then argue how the reasoning can be turned into a suitable algorithm.

Call an induced subgraph $H$ of $G$ {\em{heavy}} if $\wei(H)>(1-\alpha) \wei(G)$ and {\em{light}} otherwise.
We construct $P$ inductively so that after constructing $v_0,\ldots,v_{\ell}$, these vertices induce a path $(v_0,\ldots,v_\ell)$ in $G$ and property~\ref{p:mid} is satisfied for all $i\in \{0,1,\ldots,\ell\}$.
If no component of $G$ is heavy, we may finish the construction immediately by setting $k=-1$ and $Q$ as the empty path. Otherwise, we start by setting $v_0=u$. 
Since $G_0=G-v_0$ and $G$ is connected, the unique (due to $\alpha<1/2$) heavy component of $G_0$ is adjacent to $v_0$ and~\ref{p:mid} is satisfied for $i=0$.

For $\ell\geq 0$, the construction of $v_{\ell+1}$ is implemented as follows.
By~\ref{p:mid} for $i=\ell$, there is a connected component $D$ of $G_\ell$ that is heavy and adjacent to $v_\ell$.
As $\alpha<1/2$, no other connected component of $G_\ell$ can be heavy.
Since $G_{\ell+1}$ is an induced subgraph of $G_\ell$, either every connected component of $G_{\ell+1}$ is light, or there is exactly one heavy connected component $D'$ of $G_{\ell+1}$ that is moreover
an induced subgraph of $D$. In the former case, we may finish the construction by setting $k=\ell$, as then~\ref{p:fin} is satisfied.
Otherwise, observe that $G_{\ell+1}$ is obtained from $G_\ell$ by removing vertices of $N[v_\ell]\setminus N[v_0,\ldots,v_{\ell-1}]$,
hence $D'$ is a connected component of $D-(N[v_\ell]\cap V(D))$. Here observe that $N[v_{\ell}]\cap V(D)$ is non-empty, because $D$ is adjacent to $v_{\ell}$. 
Consequently, there exists a vertex $v_{\ell+1}\in V(D)$ that is simultaneously adjacent to $v_\ell$ and to $D'$. Since $v_{\ell+1}\in V(D)$, $v_{\ell+1}$ is not adjacent to any of the vertices $v_0,\ldots,v_{\ell-1}$.
We conclude that the induced path $(v_0,\ldots,v_{\ell})$ can be extended by $v_{\ell+1}$ so that~\ref{p:mid} is satisfied for $i=\ell+1$.

Since $G$ is finite, the construction eventually finishes yielding a path $Q$ satisfying both~\ref{p:fin} and~\ref{p:mid}. We are left with arguing the algorithmic statement.

Observe that in the above reasoning, we used the constant $\alpha$ and the function $\wei$ only in order to verify whether the construction should be finished, or to identify the heavy connected component $D'$ of
$D-(N[v_\ell]\cap V(D))$. Having identified $D'$, $v_{\ell+1}$ can be chosen freely among the common neighbors of $D'$ and $v_{\ell}$.
Fix beforehand a total order of $V(G)$ and assume that $v_{\ell+1}$ is always chosen as the smallest eligible vertex. 
Consider any run of the algorithm for $G,\alpha,\wei$ and for $i\in \{0,\ldots,k-1\}$ let $D_i$ be the unique heavy connected component of $G_i$.
Since $\alpha<1/2$, subgraphs $D_i$ pairwise intersect. Since $G_0,G_1,G_2,\ldots,G_{k-1}$ is a descending chain in the induced subgraph order and each $D_i$ is a connected component of $G_i$, we conclude that 
$D_0,D_1,D_2,\ldots,D_{k-1}$ is also a descending chain in the induced subgraph order. Consequently, there exists a vertex $z$ that is contained in each of $D_0,D_1,\ldots,D_{k-1}$.
Now comes the main observation: knowing $z$ and having constructed $G_i$, we may identify $D_i$ as the unique connected component of $G_i$ that contains $z$.
Thus, a path $Q$ suitable for $\alpha,\wei$ can be constructed knowing only $k$ and $z$ (given the total order fixed beforehand).
Constructing such a path $Q$ for every choice of $k$ and $z$, of which there are at most $\Oh(|V(G)|^2)$ many, yields the desired family $\Qq$. 
\end{proof}

Note that in the statement of Theorem~\ref{lem:gyarfas-path}, graph $G_{k+1}$ is equal to $G-N[Q]$ unless $Q$ is empty, when it is equal to $G-u$.

Now Theorem~\ref{thm:Pt-free-dispersible} follows from a straightforward combination of Lemmas~\ref{lem:guess-heavy} and Lemma~\ref{lem:gyarfas-path}.

\begin{proof}[Proof of Theorem~\ref{thm:Pt-free-dispersible}]
We first argue the $\frac{1}{2}$-uniform dispersibility. Set $\tau = \frac{1/4}{t-1}$ and assume $G$ is an $n$-vertex connected $P_t$-free graph and $|N[v]| \leq \tau \sqrt{n}$ for every $v\in V(G)$.
Apply Lemma~\ref{lem:gyarfas-path} to~$G$, arbitrary $u \in V(G)$, $\alpha = 1/4$, and uniform weight function $\wei$, obtaining a path $Q$.
Note that $Q$ cannot be empty, as $G$ is connected and the weight function is uniform.
Since $G$ is $P_t$-free, $Q$ has at most $t-1$ vertices, so $X := N[V(Q)]$ has size at most $\frac{1}{4} \sqrt{n}$. 
On the other hand, every connected component $C$ of $G-X = G_{k+1}$ has at most $(1-\alpha) |V(G)| = \frac{3}{4}n$ vertices, so
$|X| \leq n^{-\frac{1}{2}} (n - |C|)$. Hence, we can return $X$ and a trivial extended strip decomposition of $G-X$ as the desired uniform disperser.

For QP-dispersibility, the argument is only slightly longer.
Without loss of generality assume $t\geq 4$. We argue that the class of $P_t$-free graphs satisfies the prerequisites of Lemma~\ref{lem:guess-heavy}.
Thus we assume we are given a connected $P_t$-free graph $G$ and a parameter $\leps>0$.
Consider applying Lemma~\ref{lem:gyarfas-path} to $G$ and any vertex $u\in V(G)$.
We infer that in polynomial time we can construct a polynomial-size family $\Qq$ of induced paths in $G$ satisfying in particular the following:
for each weight function $\wei\colon V(G)\to \N$ there exists $Q\in \Qq$ such that $\wei(C)\leq \frac{3}{4}\wei(G)$ for every $C\in \cc(G-X)$, where $X=N[Q]$ if $Q$ is non-empty and $X=\{u\}$ otherwise.
Since $G$ is $P_t$-free, every path in $\Qq$ has less than $t$ vertices.
Consequently, supposing $\wei(N[v])\leq \frac{\leps}{4t} \cdot \wei(V(G))$ for every vertex $v$, we have $\wei(X)\leq \leps/4 \cdot \wei(V(G))$ for every $Q\in \Qq$, 
and in particular $\wei(X)\leq \leps\cdot \wei(G-C)$ for every $C\in \cc(G-X)$.

From $\Qq$ construct a family $\Nn$ by including, for every $Q\in \Qq$, a pair $(X,(H,\esd))$ where $X$ is as above and $(H,\esd)$ is the trivial extended strip decomposition of $G-X$.
The reasoning of the previous paragraph shows that the assumptions of Lemma~\ref{lem:guess-heavy} are satisfied for $p(\leps)=\frac{\leps}{4t}$.
Therefore, from Lemma~\ref{lem:guess-heavy} we conclude that the class of $P_t$-free graphs is strongly QP-dispersible.
\end{proof}

\subsection{Graphs without long holes}\label{ss:long-hole}

The proof of Theorem~\ref{thm:hole-free-dispersible} follows from applying exactly the same reasoning as in the proof of Theorem~\ref{thm:Pt-free-dispersible},
except that in order to obtain a suitable path family $\Qq$ we use the following Lemma~\ref{lem:hole-free-separator}, instead of Lemma~\ref{lem:gyarfas-path}.
Furthermore, the lemma below also directly implies Theorem~\ref{thm:hole-tw}
via standard arguments (see e.g. Corollary 1 of~\cite{BacsoLMPTL19}).

\begin{lemma}\label{lem:hole-free-separator}
Let $G$ be a connected $C_{\geq t}$-free graph supplied with a weight function $\wei\colon V(G)\to \N$.
Then in $G$ there is an induced path $Q$ on less than $t$ vertices such that
$$\wei(C)\leq \frac{3}{4}\wei(G)\quad\textrm{for every }C\in \cc(G-N[Q]).$$
Moreover, given $G$ alone, one can enumerate in polynomial time a family $\Qq$ of $\Oh(|V(G)|^2)$ induced paths on less than $t$ vertices with a guarantee that for every weight function $\wei$ there exists $Q\in \Qq$
satisfying the above for $\wei$.
\end{lemma}
\begin{proof}
Without loss of generality assume $t\geq 4$.
We first focus on proving the existential statement. At the end we will argue how the enumeration statement can be derived from the enumeration statement of Lemma~\ref{lem:gyarfas-path}.

Fix any vertex $u$ in $G$ and apply the existential statement of Lemma~\ref{lem:gyarfas-path} to $G$, vertex $u$, weight function $\wei$, and $\alpha=\frac{1}{4}$.
This yields an induced path $R=(v_0,v_1,\ldots,v_k)$ satisfying properties~\ref{p:fin} and~\ref{p:mid}, where $v_0=u$. If $k+1<t$ then, by~\ref{p:fin}, we may simply take $Q=R$, or $Q=(u)$ 
in case $R$ is the empty path. Hence, from now on assume that $k\geq t-1$.

Let $R'$ and $R''$ be the subpaths of $R$ defined as
$$R'=(v_{k-t+1},\ldots,v_{k-1})\qquad\textrm{and}\qquad R''=(v_{k-t+2},\ldots,v_k).$$
Note that each of $R',R''$ has $t-1$ vertices. 
In the rest of the proof we argue the following claim: one of paths $R',R''$ satisfies the condition required of $Q$.

Suppose otherwise: there are components $D'\in \cc(G-N[R'])$ and $D''\in \cc(G-N[R''])$ with $\wei(D')>\frac{3}{4}\wei(G)$ and $\wei(D'')>\frac{3}{4}\wei(G)$.
Note that then $D'$ and $D''$ are unique. We observe the following.

\begin{claim}\label{cl:Dp}
$D'$ is adjacent to $v_{k}$.
\end{claim}
\begin{proof}
By property~\ref{p:mid} of Lemma~\ref{lem:gyarfas-path}, $G-N[v_0,\ldots,v_{k-1}]$ contains a (unique) connected component $C$ of weight more than $\frac{3}{4}\wei(G)$ that is moreover adjacent to $v_k$.
As $G-N[v_0,\ldots,v_{k-1}]$ is an induced subgraph of $G-N[R']$ and $\wei(D')>\frac{3}{4}\wei(G)$, it follows that $C$ is contained in $D'$.
Hence $D'$ is adjacent to $v_k$.
\cqed\end{proof}

\begin{claim}\label{cl:Dpp}
$D''$ is adjacent to $v_{k-t+1}$.
\end{claim}
\begin{proof}
Graph $G-N[v_0,\ldots,v_k]$ can be obtained from $G-N[R'']$ by removing vertices $v_0,\ldots,v_{k-t}$ and all the neighbors of $v_0,\ldots,v_{k-t+1}$ that do not belong to $N[R'']$; 
denote the set of those vertices by $Z$.
Thus, every connected component of $G-N[R'']$ that is not a connected component of $G-N[v_0,\ldots,v_k]$ necessarily contains a vertex of $Z$.
Since by property~\ref{p:fin} of Lemma~\ref{lem:gyarfas-path}, no connected component of $G-N[v_0,\ldots,v_k]$ has weight more than $\frac{3}{4}\wei(G)$, while this is the case for $D''$,
we conclude that $V(D'')\cap Z\neq \emptyset$. Now observe that $G[Z\cup \{v_{k-t+1}\}]$ is connected and all vertices of $Z$ are present in $G-N[R'']$. Hence some vertex of
$V(D'')\cap Z$ is adjacent to $v_{k-t+1}$, implying the claim.
\cqed\end{proof}

\begin{claim}\label{cl:Dintersect}
$V(D')\cap V(D'')\neq \emptyset$.
\end{claim}
\begin{proof}
Follows immediately from $\wei(D')>\frac{3}{4}\wei(G)$ and $\wei(D'')>\frac{3}{4}\wei(G)$.
\cqed\end{proof}

By Claims~\ref{cl:Dp},~\ref{cl:Dpp},~\ref{cl:Dintersect} it follows that there exists an induced path $P$ with endpoints $v_{k-t+1}$ and $v_k$ whose all internal vertices belong to $V(D')\cup V(D'')$.
As vertices of $V(D')\cup V(D'')$ are non-adjacent to $v_{k-t+2},\ldots,v_{k-1}$ by definition, path $P$ together with the subpath of $R$ from $v_{k-t+2}$ to $v_{k-1}$ induce a cycle of length 
at least $t$, a contradiction.

For the enumeration statement, it suffices to compute the family $\Rr$ provided by Lemma~\ref{lem:gyarfas-path} and, for every $R\in \Rr$, include in $\Qq$ either $R$, if its number of vertices is less than $t$,
or both $R'$ and $R''$ as defined above for $R$.
\end{proof}

\section{Rooted subdivided claw}\label{sec:spider}
In this section we will focus on the classes of graphs excluding a claw subdivided a fixed number of times. 
We try to construct such subdivided claws with the use of Theorem~\ref{thm:3-in-a-tree}. 
This provides us with extended strip decompositions of the considered graphs. 

We introduce a useful lemma that encapsulates the way we will use Theorem~\ref{thm:3-in-a-tree}. We first need a definition.

\begin{definition}
Let $G$ be a graph and let $Z\subseteq V(G)$ be such that $|Z|=3$.
An extended strip decomposition $(H,\esd)$ {\em{shatters}} $Z$ if the following condition holds:
whenever $P_1,P_2,P_3$ is a triple of induced paths in $G$ that are pairwise disjoint and non-adjacent, and each of them has one endpoint in $Z$,
then there is no atom in $(H,\esd)$ that intersects or is adjacent to each of $P_1,P_2,P_3$.
\end{definition}

\begin{lemma}\label{lem:claw-shatter}
Let $G$ be a graph and let $Z\subseteq V(G)$ be such that $|Z|=3$.
Then one can in polynomial time find either an induced tree in $G$ containing all vertices of $Z$, or an extended strip decomposition $(H,\esd)$ of $G$ that shatters $Z$.
\end{lemma}

The proof of Lemma~\ref{lem:claw-shatter} is postponed to Section~\ref{sec:claw-shatter}.
Note that contrary to Theorem~\ref{thm:3-in-a-tree}, Lemma~\ref{lem:claw-shatter} does not assume that the graph is connected.

We move to the main point of this section, which concerns classes excluding subdivided claws.

\begin{definition}
A {\em{subdivided claw}} is a graph obtained from the claw $K_{1,3}$ and subdividing each of its edges an arbitrary number of times.
The degree-$1$ vertices are then called the {\em{tips}} of the claw, while the unique vertex of degree $3$ is the {\em{center}}. 
A subdivided claw is a {\em{$(\geq t)$-claw}} if all its tips are at distance at least $t$ from its center.
A graph $G$ is {\em{$Y_{\geq t}$-free}} if it does not contain any $(\geq t)$-claw as an induced subgraph.
\end{definition}

\begin{theorem}\label{thm:claw-free-dispersible}
For every $t\in \N$, the class of $Y_{\geq t}$-free graphs is QP-dispersible and $\frac{1}{9}$-uniformly dispersible.
\end{theorem}

Theorem~\ref{thm:claw-free-dispersible} is a consequence of Lemma~\ref{lem:find-claw} below.
Indeed, to obtain QP-dispersibility if suffices to combine Lemma~\ref{lem:find-claw} with Lemma~\ref{lem:guess-heavy}, 
while to obtain $\frac{1}{9}$-uniformly dispersibility, apply Lemma~\ref{lem:find-claw} for $\leps=n^{-1/9}$ (setting $n_0$ sufficiently large such that $\leps < \frac{1}{100t}$),
      uniform weight function, and any $u$, 
and observe that a pair $(X,(H,\esd))$ satisfying~\ref{c:esd} is a $\frac{1}{9}$-uniform disperser; we can find such a pair in polynomial time by inspecting all the members of $\Nn$.

\begin{lemma}\label{lem:find-claw}
Fix an integer $t\geq 4$.
Let $G$ be a connected graph supplied with a weight function $\wei\colon V(G)\to \N$ 
and let $\leps\in (0,\frac{1}{100t})$ be
such that
\begin{equation}\label{eq:light}
\wei(N[v])\leq \leps^{8}\cdot \wei(G)\textrm{ for every }v\in V(G).
\end{equation}
Let $u$ be any vertex of $G$. Then there is either
\begin{enumerate}[label=(C\arabic*),ref=(C\arabic*)]
\item\label{c:claw} an induced $(\geq t)$-claw in $G$ with one of the tips being $u$, or
\item\label{c:esd}  a subset of vertices $X\subseteq V(G)$ and an extended strip decomposition $(H,\esd)$ of $G-X$ such that 
$$\wei(A)\leq (1-\leps^7)\cdot \wei(G)\quad\textrm{and}\quad\wei(X)\leq \leps\cdot \wei(G-A)\quad\textrm{for every atom }A\textrm{ of }(H,\esd).$$
\end{enumerate}
Moreover, given $G$ and $u$ one can in polynomial time either find conclusion~\ref{c:claw}, or enumerate a family $\Nn$ of $\Oh(|V(G)|^4)$ pairs $(X,(H,\esd))$ such that
for every $\leps \in (0,\frac{1}{100t})$ and every weight function $\wei\colon V(G)\to \N$ satisfying~\eqref{eq:light} there exists $(X,(H,\esd))\in \Nn$ satisfying~\ref{c:esd} for $\wei$ and $\leps$.
\end{lemma}
\begin{proof}
We first focus on proving the existential statement. At the end we will argue how the enumeration statement can be derived using the enumeration statement of Lemma~\ref{lem:gyarfas-path}.

Apply Lemma~\ref{lem:gyarfas-path} to $G$, $u$, $\wei$, and $\alpha=\leps$, yielding a suitable path $Q=(v_0,\ldots,v_k)$, where $v_0=u$ (unless $k=-1$ and $Q$ is empty).
As in Lemma~\ref{lem:gyarfas-path}, denote $G_0=G-u$ and $G_i=G-N[v_0,\ldots,v_{i-1}]$ for $i\in \{1,\ldots,k+1\}$.
For $i\in \{0,\ldots,k+1\}$, let $D_i$ be the heaviest (w.r.t. $\wei$) connected component of $G_i$. Then by~\ref{p:mid} and~\ref{p:fin} we have
\begin{equation}\label{eq:Di}
\wei(D_i)>(1-\leps)\cdot\wei(G)\textrm{ for }i\leq k\qquad\textrm{and}\qquad \wei(D_{k+1})\leq (1-\leps)\cdot\wei(G).
\end{equation}
Also, as argued in the proof of Lemma~\ref{lem:gyarfas-path}, $D_j$ is an induced subgraph of $D_i$ for each $i,j\in \{0,\ldots,k\}$ with $i\leq j$.

If $\wei(D_0)\leq (1-\leps^5)\cdot \wei(G)$, then conclusion~\ref{c:esd} can be obtained by taking $X=\{v_0\}$ and $(H,\esd)$ to be the trivial extended strip decomposition of $G-X$.
This is because $\wei(X)=\wei(v_0)\leq \leps^{8}\cdot \wei(G)$ due to~\eqref{eq:light}, while $\wei(G-D)\geq \leps^5\cdot \wei(G)$ for every connected component $D$ of $G-X$.
Note that if $k=-1$, then in particular $\wei(D_0)\leq (1-\leps)\cdot \wei(G)\leq (1-\leps^5)\cdot \wei(G)$, so the above analysis can be applied as well.
Hence, from now on assume that $k\geq 0$ and $\wei(D_0)>(1-\leps^5)\cdot \wei(G)$.

Define $p$ and $q$ as the largest indices satisfying the following:
$$\wei(D_p)>(1-\leps^5)\cdot\wei(G)\qquad\textrm{and}\qquad\wei(D_q)>(1-\leps^3)\cdot \wei(G).$$
By~\eqref{eq:Di} and the discussion of the previous paragraph we have that $p$ and $q$ are well-defined and satisfy $0\leq p\leq q\leq k$.

We now observe that indices $0,p,q,k$ have to be well-separated from each other, or otherwise we are done.
For this, consider the following paths in $G$:
$$R_1=(v_0,v_1,\ldots,v_{p-2}),\qquad R_2=(v_p,v_{p+1},\ldots,v_{q-2}),\qquad R_3=(v_q,v_{q+1},\ldots,v_{k-1}).$$
Note that the above path formally may be empty in case the index of the second endpoint is smaller than that of the first endpoint; in a moment we will see that this is actually never the case.
We now verify that the neighborhood of each of these paths has to have a significant weight, or otherwise we are done.

\begin{claim}\label{cl:paths-heavy}
If we have
$$\wei(N[R_1])\leq (\leps^6/2) \cdot \wei(G)\quad\textrm{or}\quad\wei(N[R_2])\leq (\leps^4/2) \cdot \wei(G)\quad\textrm{or}\quad\wei(N[R_3])\leq (\leps^2/2) \cdot \wei(G),$$
then conclusion~\ref{c:esd} can be obtained.
\end{claim}
\begin{proof}
We first consider the case when $\wei(N[R_1])\leq (\leps^6/2) \cdot \wei(G)$, which is slightly simpler.
By assumption we have $\wei(D_{p+1})\leq (1-\leps^5)\cdot\wei(G)$ where $D_{p+1}$ is the heaviest connected component of $G-N[v_0,v_1,\ldots,v_p]$.
On the other hand, we have
\begin{eqnarray*}
\wei(N[v_0,v_1,\ldots,v_p]) & \leq & \wei(N[R_1])+\wei(N[v_{p-1}])+\wei(N[v_{p}])\\
                            & \leq & (\leps^6/2 + 2\leps^8)\cdot \wei(G)\\
                            & \leq & \leps^6\cdot \wei(G).
\end{eqnarray*}
Hence, we can obtain conclusion~\ref{c:esd} by taking $X=N[v_0,v_1,\ldots,v_p]$ and the trivial extended strip decomposition of $G-X$.
Indeed, for every connected component $D$ of $G-X$ we have $\wei(D)\leq \wei(D_{p+1})\leq (1-\leps^5)\cdot\wei(G)$, implying also that
$\wei(X)\leq \leps^6\cdot \wei(G)\leq \leps\cdot \wei(G-D)$.

Now, consider the case when $\wei(N[R_2])\leq (\leps^4/2)\cdot \wei(G)$.
Observe that we also have $\wei(N[R_1])\leq \wei(G)-\wei(D_p)<\leps^5\cdot \wei(G)$, because $D_p$ and $N[R_1]$ are disjoint.
By assumption we have $\wei(D_{q+1})\leq (1-\leps^3)\cdot\wei(G)$ where $D_{q+1}$ is the heaviest connected component of $G-N[v_0,v_1,\ldots,v_q]$.
On the other hand, we have
\begin{eqnarray*}
\wei(N[v_0,v_1,\ldots,v_q]) & \leq & \wei(N[R_1])+\wei(N[R_2])+\wei(N[v_{p-1}])+\wei(N[v_{q-1}]) \\
                            & \leq & (\leps^5 + \leps^4/2 + 2\leps^8)\cdot \wei(G) \\
                            & \leq & \leps^4\cdot \wei(G).
\end{eqnarray*}
Hence, we can obtain conclusion~\ref{c:esd} by taking $X=N[v_0,v_1,\ldots,v_q]$ and the trivial extended strip decomposition of $G-X$.
Indeed, for every connected component $D$ of $G-X$ we have $\wei(D)\leq \wei(D_{q+1})\leq (1-\leps^3)\cdot\wei(G)$, implying also that
$\wei(X)\leq \leps^4\cdot \wei(G)\leq \leps\cdot \wei(G-D)$.

Finally, consider the case when $\wei(N[R_3])\leq (\leps^2/2) \cdot \wei(G)$.
As in the previous case, we have $\wei(N[R_1])<\leps^5\cdot \wei(G)$ and $\wei(N[R_2])<\leps^3\cdot \wei(G)$ .
By the construction of $Q$ we have $\wei(D_{k+1})\leq (1-\leps)\cdot\wei(G)$ where $D_{k+1}$ is the heaviest connected component of $G-N[v_0,v_1,\ldots,v_k]$.
On the other hand, we have
\begin{eqnarray*}
\wei(N[v_0,v_1,\ldots,v_k]) & \leq & \wei(N[R_1])+\wei(N[R_2])+\wei(N[R_3]+\wei(N[v_{p-1}])+\wei(N[v_{q-1}])+\wei(N[v_k]) \\
                            & \leq & (\leps^5 + \leps^3 + \leps^2/2 + 3\leps^8)\cdot \wei(G) \\
                            & \leq & \leps^2\cdot \wei(G).
\end{eqnarray*}
Hence, we can obtain conclusion~\ref{c:esd} by taking $X=N[v_0,v_1,\ldots,v_k]$ and the trivial extended strip decomposition of $G-X$.
Indeed, for every connected component $D$ of $G-X$ we have $\wei(D)\leq \wei(D_{k+1})\leq (1-\leps)\cdot\wei(G)$, implying also that
$\wei(X)\leq \leps^2\cdot \wei(G)\leq \leps\cdot \wei(G-D)$.
\cqed\end{proof}

We proceed under the assumption that the prerequisite of Claim~\ref{cl:paths-heavy} does not hold, that is,
\begin{equation}\label{eq:paths-heavy}
\wei(N[R_1])>(\leps^6/2) \cdot \wei(G)\textrm{ and }\wei(N[R_2])>(\leps^4/2) \cdot \wei(G)\textrm{ and }\wei(N[R_3])>(\leps^2/2) \cdot \wei(G).
\end{equation}
From this we argue that $0,p,q,k$ have to be well-separated from each other.

\begin{claim}\label{cl:close-done}
It holds that
$$p-0>t+1\qquad\textrm{and}\qquad q-p>t+1\qquad\textrm{and}\qquad k-q>t+1.$$
\end{claim}
\begin{proof}
Observe that if $p-0\leq t+1$, then 
$$\wei(N[R_1])\leq \sum_{i=0}^{p-1}\wei(N[v_i])\leq (t+1)\leps^8\cdot \wei(G)<(\leps^6/2)\cdot \wei(G),$$
contradicting the assumption~\eqref{eq:paths-heavy}. The proof for the other two inequalities is analogous.
\cqed\end{proof}

We will also consider the following subpaths of $Q$:
$$Q_1=(v_0,v_1,\ldots,v_{t-1}),\qquad Q_2=(v_p,v_{p+1},\ldots,v_{p+t-1}),\qquad Q_3=(v_q,v_{q+1},\ldots,v_{q+t-1}).$$
Note that by Claim~\ref{cl:close-done}, paths $Q_1,Q_2,Q_3$ are pairwise disjoint and non-adjacent, and they are prefixes of $R_1,R_2,R_3$, respectively. Also, each of them consists of $t$ vertices.

Now, let
$$G'=G-((N(Q_1)\cup N(Q_2)\cup N(Q_3))\setminus \{v_t,v_{p+t},v_{q+t}\}).$$
Note that in $G'$, paths $Q_1,Q_2,Q_3$ are preserved, but they become {\em{detached}} in the following sense: only one endpoint ($v_{t-1},v_{p+t-1},v_{q+t-1}$, respectively) is adjacent to one vertex from the rest
of the graph ($v_t,v_{p+t},v_{q+t}$, respectively). Also, paths $R_1,R_2,R_3$ are also preserved in $G'$.

We now apply Lemma~\ref{lem:claw-shatter} to graph $G'$ with
$$Z=\{v_0,v_p,v_q\}.$$
This either yields an induced tree $T$ in $G'$ containing $v_0,v_p,v_q$, or an extended strip decomposition $(H',\esd')$ of $G'$ which shatters $v_0,v_p,v_q$.
In the first case, by the construction of $G'$ it follows that $T$ has to contain an induced $(\geq t)$-claw $T'$ with tips $v_0,v_p,v_q$. As $v_0=u$, then $T'$ witnesses that conclusion~\ref{c:claw} holds.
Hence, from now on we assume the second case.

Observe that
$$\wei(N[Q_1]\cup N[Q_2]\cup N[Q_3])\leq 3t\cdot \leps^8\cdot \wei(G)\leq (\leps^7/2) \cdot \wei(G).$$
Hence, it now suffices to prove the following:
\begin{equation}\label{eq:atoms-shattered}
\wei(A)\leq (1-\leps^6/2)\cdot \wei(G)\textrm{ for every atom }A\textrm{ of }(H,\esd).
\end{equation}
Indeed, if~\eqref{eq:atoms-shattered} holds,
then we can obtain conclusion~\ref{c:esd} by taking $X=N[Q_1]\cup N[Q_2]\cup N[Q_3]$ and $(H,\esd)$ to be $(H',\esd')$ with all the vertices of $V(Q_1)\cup V(Q_2)\cup V(Q_3)\cup \{v_t,v_{p+t},v_{q+t}\}$ removed,
because then
$$\wei(X)\leq (\leps^7/2) \cdot \wei(G)\leq \leps\cdot \wei(G-A)\textrm{ for every atom }A\textrm{ of }(H,\esd).$$

Suppose that, contrary to~\eqref{eq:atoms-shattered}, there exists an atom $A$ in $(H',\esd')$ such that $\wei(A)>(1-\leps^6/2)\cdot \wei(G)$.
Note that since $Q$ is an induced path in $G$, we have that $R_1,R_2,R_3$ are induced paths in $G'$ that are disjoint and pairwise non-adjacent. 
Since $(H',\esd')$ shatters $\{v_0,v_p,v_q\}$, we conclude that the atom $A$ is disjoint with $N[R_t]$ for at least one $t\in \{1,2,3\}$.
However, this combined with~\eqref{eq:paths-heavy} and the assumption that $\wei(A)>(1-\leps^6/2)\cdot \wei(G)$ yields that $\wei(A\cup N[R_t])>\wei(G)$, a contradiction.
This concludes the proof of the existential statement.

\medskip

For the enumeration statement, it suffices to enumerate the family $\Qq$ provided by Lemma~\ref{lem:gyarfas-path}, and for every $Q=(v_0,\ldots,v_k)$ and $0\leq p\leq q\leq k$ include in $\Nn$ the following pairs:
\begin{itemize}
\item $X=N[v_0,\ldots,v_p]$, and the trivial extended strip decomposition of $G-X$;
\item $X=N[v_0,\ldots,v_q]$, and the trivial extended strip decomposition of $G-X$;
\item $X=N[v_0,\ldots,v_k]$, and the trivial extended strip decomposition of $G-X$;
\item $X=N[v_0,\ldots,v_{t-1}]\cup N[v_p,\ldots,v_{p+t-1}]\cup N[v_q,\ldots,v_{q+t-1}]$, and the extended strip decomposition obtained by applying Theorem~\ref{thm:3-in-a-tree} to $G'$ 
(in the notation from the proof above) and $Z=\{v_0,v_p,v_q\}$. 
\end{itemize}
In the last point, if for any choice of $Q,p,q$ we obtain an induced $(\geq t)$-claw with $u$ as one of the tips, then it can be reported by the algorithm.
Otherwise from the above proof it is clear that the enumerated family $\Nn$ consists of $\Oh(|V(G)|^4)$ pairs and satisfies the required property.
\end{proof}

\subsection{Proof of Lemma~\ref{lem:claw-shatter}}\label{sec:claw-shatter}

The following technical lemma describes how triples of disjoint, non-adjacent paths starting at peripheral vertices behave in an extended strip decomposition of a graph.

\begin{lemma}\label{lem:esd-atoms-protected}
Let $(H,\esd)$ be an extended strip decomposition of a graph $G$.
Suppose $P_1,P_2,P_3$ are three induced paths in $G$ that are pairwise disjoint and non-adjacent, and moreover each of $P_1,P_2,P_3$ has an endpoint that is peripheral in $(H,\esd)$.
Then in $(H,\esd)$ there is no atom that would intersect or be adjacent to each of $P_1,P_2,P_3$.
\end{lemma}
\begin{proof}
A {\em{feature}} of $(H,\esd)$ is a vertex, an edge, or a triangle of $H$.
We introduce the following incidence relation between features: two edges are incident if they share a vertex, a vertex of $H$ is incident to all edges of $H$ it is an endpoint of, and a triangle of $H$ is incident to all edges of $H$ that it contains.
Thus, vertices and triangles are considered to be non-incident. Note that every edge of $G$ connects either vertices from $\esd(f)$ for the same feature $f$, or from $\esd(f)$ and $\esd(f')$ for two incident
features $f,f'$.

Consider an induced path $Q$ in $G$. 
A {\em{visit}} of a feature $f$ by $Q$ is a maximal subpath of $Q$ consisting of vertices belonging to $\esd(f)$.
The order of vertices on $Q$ naturally gives rise to an order of visits of features by $Q$.
We now establish a few basic properties of how induced paths in $G$ behave w.r.t. the decomposition $(H,\esd)$ in order to get an understanding of the interaction between $P_1,P_2,P_3$ in $(H,\esd)$.

\begin{claim}\label{cl:hobbit}
Suppose $Q$ is an induced path in $G$. Consider some visit $W$ of a feature $f$ by $Q$, where $f$ is either a vertex or a triangle.
Let $W_1$ be the visit on $Q$ directly before $W$ and $W_2$ be the visit on $Q$ directly after $W$; possibly $W_1$ or $W_2$ does not exist when $W$ is the first, respectively last visit of a feature on $Q$.
Then $W_1$ and $W_2$, if existent, are visits of an edge in $H$ that is incident to $f$, and if they are both existent, then this is the same edge of $H$.
\end{claim}
\begin{proof}
Let $f_1,f_2$ be the features visited by $Q$ in $W_1,W_2$, respectively.
The fact that $f_1,f_2$ are both edges incident to $f$ follows directly from the definition of an extended strip decomposition, in particular the conditions on edges of $G$.
We are left with proving that if both $W_1,W_2$ exist (i.e., visit $W$ appears neither at the front nor at the end of $Q$), then $f_1=f_2$.

Consider first the case when $f$ is a vertex. 
Then $f_1$ and $f_2$ are both edges incident to $f$.
Moreover, then the last vertex of the visit $W_1$ belongs to $\esd(f_1,f)$, while the first vertex of the visit of $W_2$ belongs to $\esd(f_2,f)$.
But if $f_1\neq f_2$, then $\esd(f_1,f)$ and $\esd(f_2,f)$ would be complete to each other, which would contradict the assumption that $P$ is induced.
Therefore we conclude that $f_1=f_2$.

Consider now the case when $f$ is a triangle; then $f_1$ and $f_2$ are both edges contained in $f$.
Supposing $f_1\neq f_2$, we may denote $f=uvw$, $f_1=uv$, $f_2=uw$.
Then the last vertex of the visit $W_1$ belongs to $\esd(uv,u)\cap \esd(uv,v)$, while the first vertex of the visit of $W_2$ belongs to $\esd(uw,u)\cap \esd(uw,w)$.
This means that these two vertices are adjacent, because they belong to $\esd(uv,u)$ and $\esd(uw,u)$, respectively.
This is a contradiction with the assumption that $P$ is an induced path.
\cqed\end{proof}

\begin{claim}\label{cl:disjoint-esd}
Suppose $Q_1$ and $Q_2$ are two induced paths in $G$ that do not intersect and are non-adjacent.
Suppose further that $Q_1$ has endpoint $z_1$ and $Q_2$ has endpoint $z_2$ such that $z_1,z_2$ are peripheral.
Then there does not exist an edge $uv$ of $H$ such that both $Q_1$ and $Q_2$ intersect $\esd(uv,u)$.
\end{claim}
\begin{proof}
Orient $Q_1,Q_2$ so that $z_1,z_2$ are their first vertices, respectively.
Suppose the claim does not hold and let $(uv,u)$ be such that both $Q_1$ and $Q_2$ intersect $\esd(uv,u)$; among such pairs, choose $(uv,u)$ so that the distance from $z_1$ to the first vertex of $\esd(uv,u)$ on $Q_1$ 
plus the distance from $z_2$ to the first vertex of $\esd(uv,u)$ on $Q_2$ 
is as small as possible.
Let $y_1,y_2$ be the first vertices on $Q_1,Q_2$ that belong to $\esd(uv,u)$, respectively.

Consider first the corner case when $z_1=y_1$ and $z_2=y_2$.
Since both $z_1,z_2$ are peripheral and $z_1,z_2\in \esd(uv)$, it must be that $\esd(uv,u)=\{z_1\}$ and $\esd(uv,v)=\{z_2\}$, or vice versa.
But then $z_2\notin \esd(uv,u)$, a contradiction.

Hence, either $y_1\neq z_1$ or $y_2\neq z_2$.
Assume without loss of generality the former and let $x_1$ be the vertex directly preceding $y_1$ on $Q_1$; clearly, $x_1\notin \esd(uv,u)$ by the choice of $y_1$.

First observe that $x_1$ cannot belong to ($\esd(\cdot)$ of) any vertex or triangle of $H$.
Indeed, if this was the case, then by Claim~\ref{cl:hobbit} we would conclude that $Q_1$ would already intersect $\esd(uv,u)$ before $x_1$, so $y_1$ would not be the first vertex of $\esd(uv,u)$ on $Q_1$.
Hence, either $x_1\in \esd(uw,u)$ for some $w\neq v$, or $x_1\in \esd(uv)\setminus \esd(uv,u)$.
In the former case we infer that $x_1$ and $y_2$ would be adjacent, a contradiction with the assumption that $Q_1$ and $Q_2$ are non-adjacent.
Hence, we have $x_1\in \esd(uv)\setminus \esd(uv,u)$.
Since $Q_1$ starts in a peripheral vertex $z_1$, we conclude that on $Q_1$ there is a vertex $t_1\in \esd(uv,v)$ that appears no later than $x_1$ (possibly $t_1=x_1$).

Consider now the corner case when $z_2=y_2$. Let $ww'\in E(H)$ be such that $w$ has degree $1$ in $H$ and $\esd(ww',w)=\{z_2\}$.
Then $(uv,u)=(ww',w)$ or $(uv,u)=(ww',w')$.
In the former case we would have $y_1\in \esd(ww',w)$ and $y_1\neq y_2=z_2$, a contradiction to $|\esd(ww',w)|=1$.
In the latter case, however, we would have $t_1\in \esd(ww',w)$, again a contradiction to $|\esd(ww',w)|=1$, because $t_1\neq z_2$. 

Hence, from now on assume that $z_2\neq y_2$. By applying the same reasoning to $Q_2$ as we did for $Q_1$ we infer that on $Q_2$ there is a vertex $t_2\in \esd(uv,v)$ that appears earlier than $y_2$.
However, now the existence of $t_1,t_2\in \esd(uv,v)$ is a contradiction with the choice of the pair $(uv,u)$.
\cqed\end{proof}

We proceed to the proof of the lemma statement.
It suffices to prove the statement for atoms of the form $A^{uv}_e$ for some edge $e=uv\in E(H)$, as every atom of $(H,\esd)$ is contained in an atom of this form, 
apart from atoms corresponding to isolated vertices of $H$ for which the statement holds trivially.
Recall that then 
$$A^{uv}_e = \esd(u)\cup \esd(v)\cup \esd(uv)\cup \bigcup_{T\supseteq uv} \esd(T).$$
We first note the following.

\begin{claim}\label{cl:entry-ruined}
Among paths $P_1,P_2,P_3$, at most one can intersect the set $\esd(u)\cup \bigcup_{w\colon uw\in E(H)} \esd(uw,u)$.
\end{claim}
\begin{proof}
As each of the paths $P_1,P_2,P_3$ starts in a peripheral vertex, intersecting $\esd(u)$ entails intersecting $\bigcup_{w\colon uw\in E(H)} \esd(uw,u)$.
By Claim~\ref{cl:disjoint-esd}, no two of the paths $P_1,P_2,P_3$ intersect the same set $\esd(uw,u)$, for some $w$ with $uw\in E(H)$.
However, if, say, $P_1$ intersected $\esd(uw_1,u)$ and $P_2$ intersected $\esd(uw_2,u)$ for some $uw_1,uw_2\in E(H)$, $w_1\neq w_2$, then $P_1$ and $P_2$ would contain adjacent vertices, a contradiction.
\cqed\end{proof}

Denote
\begin{eqnarray*}
K_u & = & \esd(u)\cup \bigcup_{w\colon uw\in E(H)} \esd(uw,u),\\
K_v & = & \esd(v)\cup \bigcup_{w\colon vw\in E(H)} \esd(vw,v),\\
L   & = & \bigcup_{T\supseteq uv} \esd(T)\cup (\esd(uv)\setminus (\esd(uv,u)\cup \esd(uv,v))),
\end{eqnarray*}
and observe that
$$N[A^{uv}_e] = K_u\cup K_v\cup L.$$
By Claim~\ref{cl:entry-ruined}, $K_u$ above can be intersected by at most one of the paths $P_1,P_2,P_3$, and similarly for $K_v$.
Hence, if $N[A^{uv}_e]$ is intersected by all three paths $P_1,P_2,P_3$, then one of them, say $P_3$, intersects $L$ 
while not intersecting $K_u\cup K_v$.
Note that $$N(L)\subseteq \esd(uv,u)\cup \esd(uv,v)\subseteq K_u\cup K_v,$$ hence we conclude that $P_3$ 
is entirely contained in $L$. This is a contradiction with the assumption that one of the endpoints of $P_3$ is peripheral in $(H,\esd)$.
\end{proof}

The proof of Lemma~\ref{lem:claw-shatter} is now an easy combination of Theorem~\ref{thm:3-in-a-tree} and 
Lemma~\ref{lem:esd-atoms-protected}.

\begin{proof}[Proof of Lemma~\ref{lem:claw-shatter}.]
Consider first the case when vertices of $Z$ are not all in the same connected component of $G$. Then we can output the trivial extended strip decomposition of $G$, as it clearly shatters $Z$.

Suppose now that all vertices of $Z$ are in the same connected component $C$ of $G$.
Apply Theorem~\ref{thm:3-in-a-tree} to $Z$ in $C$.
Then, in polynomial time we can either find an induced tree $T$ in $C$ that contains all vertices of $Z$ (which can be directly returned), or an extended strip decomposition $(H_C,\esd_C)$ of $C$ such that all vertices of $Z$ are peripheral
in $(H_C,\esd_C)$. 
In the latter case, by Lemma~\ref{lem:esd-atoms-protected} we conclude that $Z$ is shattered by $(H_C,\esd_C)$ in $C$.
We augment $(H_C,\esd_C)$ to an extended strip decomposition $(H,\esd)$ of $G$ by adding for every component $C'\in \cc(G)$, $C'\neq C$, a new isolated vertex $v_{C'}$ with $\esd(v_{C'})=V(C')$.
Then it is easy to see that $(H,\esd)$ shatters $Z$ in $G$.
\end{proof}

\subsection{Proof of Theorems~\ref{thm:main} and~\ref{thm:main2}}\label{ssec:thm:main}

With Theorem~\ref{thm:claw-free-dispersible}, the proofs of Theorems~\ref{thm:main} and~\ref{thm:main2} are straightforward.
\begin{proof}[Proof of Theorem~\ref{thm:main}.]
Let $H$ be such that every connected component of $H$ is a path or a subdivided claw. 
Let $Y$ be a subdivided claw such that every connected component of $H$ is an induced subgraph of $Y$.

Let $G$ be $H$-free, let $\wei \colon V(G) \to \mathbb{N}$ be a weight function, and let $\varepsilon > 0$ be an accuracy parameter.
Set $\beta := \varepsilon / (2|V(H)|)$. 
Let $I$ be an independent set in $(G,\wei)$ of  maximum-weight.
By Lemma~\ref{lem:heavy}, there exists a set $J \subseteq I$ of size at most $\lceil \beta \log n \rceil = \Oh(\varepsilon^{-1} \log n)$
such that all $\beta$-heavy vertices w.r.t. $I$ are contained in $N[J]$. 
By branching into $n^{\Oh(\varepsilon^{-1} \log n)}$ subcases, we guess the set $J$.

Let $G' = G - N[J]$. 
Let $\mathcal{C}$ be a maximal family of connected components of $H$ such that $H[\bigcup \mathcal{C}]$ is an induced
subgraph of $G'$. Let $H' = H[\bigcup \mathcal{C}]$ and note that $H'$ is a proper induced subgraph of $H$.
Let $X \subseteq V(G')$ be such that $G'[X]$ is isomorphic to $H'$. 
Note that $|X| < |V(H)|$.

Observe that $G'' := G'-N[X]$ is $Y$-free. Indeed, if $G''$ contains $Y$ as an induced subgraph, then, by the choice of $Y$, it contains
some connected component $C$ of $H-V(H')$ as an induced subgraph. Together with $G'[X]$ isomorphic to $H'$, this contradicts the choice of $\mathcal{C}$.

Apply the algorithm of Theorem~\ref{thm:claw-free-dispersible} to find an independent set $I''$ in $G''$
that is a $(1-\varepsilon/2)$-approximation to a maximum weight independent set problem on $G''$ and $\wei|_{V(G'')}$.
This takes time $2^{\mathrm{poly}(\varepsilon^{-1}, \log n)}$ and we have $\wei(I'') \geq (1-\varepsilon/2) \wei(I \cap V(G''))$.
Finally, we return $I' := I'' \cup J$. 

Consider the branch where $J$ is guessed correctly.
We have $\wei(I \cap N[X]) \leq \beta |X| \wei(I) < \varepsilon/2 \cdot \wei(I)$.
Furthermore, 
  $$\wei(I) - \wei(I'') \leq (\varepsilon/2) \cdot \wei(I \cap V(G'')) + \wei(I \cap N[X]) \leq \varepsilon \wei(I).$$
This finishes the proof of Theorem~\ref{thm:main}.
\end{proof}

\begin{proof}[Proof of Theorem~\ref{thm:main2}.]
Again, let $H$ be such that every connected component of $H$ is a path or a subdivided claw and
let $Y$ be a subdivided claw such that every connected component of $H$ is an induced subgraph of $Y$.

Let $G$ be $H$-free and let $\wei \colon V(G) \to \mathbb{N}$ be a weight function.
If $G$ is disconnected, recurse on every connected component separately.
If $G$ contains a vertex $v$ such that $|N[v]| > |V(G)|^{1/9}$, branch exhaustively on $v$:
in one branch, delete $v$ and recurse, in the other branch, delete $N[v]$, recurse, and add $v$ to the independent set
returned by the recursive call; finally, output the one of the two obtained independent sets of higher weight. 

Otherwise, let $\mathcal{C}$ be a maximal family of connected components of $H$ such that $H[\bigcup \mathcal{C}]$ is an induced
subgraph of $G$. Let $H' = H[\bigcup \mathcal{C}]$ and note that $H'$ is a proper induced subgraph of $H$.
Let $X \subseteq V(G)$ be such that $G[X]$ is isomorphic to $H'$. 
Note that $|X| < |V(H)|$.

Observe that $G-N[X]$ is $Y$-free. Indeed, if $G-N[X]$ contains $Y$ as an induced subgraph, then, by the choice of $Y$, it contains
some connected component $C$ of $H-V(H')$ as an induced subgraph. Together with $G[X]$ isomorphic to $H'$, this contradicts the choice of $\mathcal{C}$.

For every independent set $Z \subseteq N[X]$, invoke the algorithm of Theorem~\ref{thm:claw-free-dispersible} on the $Y$-free graph
$G-(N[X]\cup N[Z])$, obtaining an independent set $I(Z)$, and observe that $I_Z := Z \cup I(Z)$ is an independent set in $G$.
Out of all independent sets $I_Z$ for $Z \subseteq N[X]$, return the one of maximum weight. 

Since we consider every independent set $Z \subseteq N[X]$, the returned solution is indeed an independent set in $G$ of maximum possible weight.

For the running time bound, note that $|N[X]| < |V(G)|^{1/9} \cdot |V(H)|$, hence we invoke
less than $2^{|V(G)|^{1/9} \cdot |V(H)|}$ calls to the algorithm of Theorem~\ref{thm:claw-free-dispersible}, each
taking $2^{\Oh(|V(G)|^{8/9} \log |V(G)|)}$ time. 
In a recursive step, the analysis is straightforward if $G$ is disconnected
and follows by standard arguments if $G$ contains a vertex $v$ with $|N[v]| > |V(G)|^{1/9}$.
This finishes the proof of Theorem~\ref{thm:main2}.
\end{proof}

\section{A small generalization}\label{sec:horse}
In this section we generalize Theorem~\ref{thm:main} by proving that Conjecture~\ref{conj:main} holds for all subcubic forests $H$ that have at most three vertices of degree three.
Let $L$ be the {\em{lobster graph}} depicted in Figure~\ref{fig:lobster}.
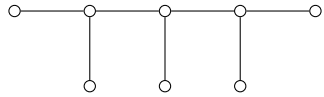
\begin{figure}[h]
\begin{center}
\begin{tikzpicture}
   \tikzstyle{vertex}=[circle,draw=black,fill=white,minimum size=0.15cm,inner sep=0pt]
        
       \foreach \c in {0,1,2,3,4} {
	  \node[vertex] (a\c) at (\c,1) {};
	}
       \foreach \c in {1,2,3} {
	  \node[vertex] (b\c) at (\c,0) {};
	  \draw (a\c) -- (b\c);
	}
       \foreach \c/\d in {0/1,1/2,2/3,3/4} {
	  \draw (a\c) -- (a\d);
	}
\end{tikzpicture}
\end{center}
\caption{The lobster graph $L$.}\label{fig:lobster}
\end{figure}
For $t\in \N$, an {\em{$(\geq t)$-lobster}} is any graph obtained from $L$ by subdividing every edge at least $t-1$ times.
Then a graph is {\em{$L_{\geq t}$-free}} if it does not contain any $(\geq t)$-lobster as an induced subgraph.

By Theorem~\ref{thm:disperser}, to prove Conjecture~\ref{conj:main} for all subcubic forests $H$ that have at most three vertices of degree three it suffices to prove the following.

\begin{theorem}\label{thm:lobster-free-dispersible}
For every $t\in \N$, the class of $L_{\geq t}$-free graphs is QP-dispersible and $\frac{1}{41}$-uniformly dispersible.
\end{theorem}

Similarly as was the case with Theorem~\ref{thm:claw-free-dispersible} and Lemma~\ref{lem:find-claw}, to prove Theorem~\ref{thm:lobster-free-dispersible} it suffices to show the following.

\begin{lemma}\label{lem:find-lobster}
Fix an integer $t\geq 4$.
Let $G$ be a connected graph supplied with a weight function $\wei\colon V(G)\to \N$
and let $\leps\in (0,\frac{1}{100t})$ be
such that
\begin{equation}\label{eq:light2}
\wei(N[v])\leq \leps^{40}\cdot \wei(G)\textrm{ for every }v\in V(G).
\end{equation}
Then there is either
\begin{enumerate}[label=(L\arabic*),ref=(L\arabic*)]
\item\label{c:lobster} an induced $(\geq t)$-lobster in $G$, or
\item\label{c:lesd}  a subset of vertices $X\subseteq V(G)$ and an extended strip decomposition $(H,\esd)$ of $G-X$ such that 
$$\wei(A)\leq (1-\leps^{39})\cdot \wei(G)\quad\textrm{and}\quad\wei(X)\leq \leps\cdot \wei(G-A)\quad\textrm{for every atom }A\textrm{ of }(H,\esd).$$
\end{enumerate}
Moreover, given $G$ one can in polynomial time either find conclusion~\ref{c:lobster}, or enumerate a family $\Nn$ of $\Oh(|V(G)|^{12})$ pairs $(X,(H,\esd))$ such that
for every $\leps \in (0,\frac{1}{100t})$ and every weight function $\wei\colon V(G)\to \N$ satisfying~\eqref{eq:light2}
there exists $(X,(H,\esd))\in \Nn$ satisfying~\ref{c:lesd} for $\wei$ and $\leps$.
\end{lemma}

The proof of Lemma~\ref{lem:find-lobster} uses the same set of ideas as that of Lemma~\ref{lem:find-claw}, but the number of steps in the construction of a lobster is larger and one needs to tend to more technical details. Essentially, the overall strategy can be summarized as follows. 
We try to construct an induced $(\geq t)$-lobster in $G$; each step of the construction may fail and produce conclusion~\ref{c:lesd} as a result. 
We start by building the right claw $T$ of the lobster using Lemma~\ref{lem:find-claw}, however we make sure that one of the tips of this claw, call it $w$, is adjacent to a connected component of $G-N[T-w]$
that contains almost the whole weight of the graph. This is done by applying Lemma~\ref{lem:gyarfas-path} to construct a long Gy\'arf\'as path $Q$, and then applying Lemma~\ref{lem:find-claw} not to any initial vertex, but to a vertex $v_i$ of the Gy\'arf\'as path such that $\wei(G_i)$ is significantly separated from $\wei(G)$. Having constructed $T$ and $w$, we forget about the first Gy\'arf\'as path $Q$ and construct, using Lemma~\ref{lem:gyarfas-path}, a second Gy\'arf\'as path $P$, this time starting from $w$. 
We construct the left claw $S$ of the lobster, but again we start this construction at later sections of $P$ so that we can ensure the following: 
there is a tip $v$ of $S$ so that in the graph $G-N[S-v]-N[T-w]$ there is a connected component containing $v$, $w$, and a long prefix of $P$. 
Then we construct the ``tail'' (that is, the middle pendant edge) of the lobster from the saved prefix of $P$, by applying Lemma~\ref{lem:claw-shatter} in this component in a manner similar to how we did it in the proof of Lemma~\ref{lem:find-claw}.

We now proceed to the formal details.

\begin{proof}[Proof of Lemma~\ref{lem:find-lobster}]
As usual, we first focus on proving the existential statement, and at the end we argue how the proof can be turned into an enumeration algorithm.

Let a {\em{$t$-claw}} be a subdivided claw in which all the tips are at distance exactly $t$ from the center. Note that a $t$-claw has exactly $3t+1$ vertices.
The first step is to use Lemmas~\ref{lem:gyarfas-path} and~\ref{lem:find-claw} to find an induced $t$-claw in $G$ that is placed robustly with respect to further constructions.

\begin{claim}\label{cl:right-claw}
We can either reach conclusion~\ref{c:lesd}, or find an induced $t$-claw $T$ in $G$ with
a tip $w$ satisfying the following property:
there is a connected component $D$ of the graph $G-N[V(T)\setminus \{w\}]$ that is adjacent to $w$ and satisfies $\wei(D)\geq (1-\leps^{35})\cdot \wei(G)$.
\end{claim}
\begin{proof}
Pick any vertex $u$ of $G$ and apply Lemma~\ref{lem:gyarfas-path} to $G$, $u$, and $\alpha=\leps$. 
This yields a suitable induced path $Q=(v_0,v_1,\ldots,v_k)$ in $G$, where $v_0=u$.
We adopt the notation from Lemma~\ref{lem:gyarfas-path} and let $D_i$ be the heaviest connected component of $G_i$, for $i\in \{0,1,\ldots,k+1\}$.
As in the proof of Lemma~\ref{lem:find-claw}, we have that $\wei(D_i)>(1-\leps)\cdot\wei(G)$ for all $i\leq k$ and $\wei(D_{k+1})\leq (1-\leps)\cdot \wei(G)$.
In the same manner as in the proof of Lemma~\ref{lem:find-claw}, we may assume that $\wei(D_0)>(1-\leps^{39})\cdot \wei(G)$, which in particular entails $k\geq 0$, for otherwise conclusion~\ref{c:lesd} can be immediately reached by taking $X=\{u\}$ and the trivial extended strip decomposition of $G_0=G-u$.

We now define $p$ as the largest index satisfying the following:
$$\wei(D_p)>(1-\leps^{35})\cdot\wei(G).$$
Since $\wei(D_0)>(1-\leps^{39})\cdot \wei(G)$ and $\wei(D_{k+1})\leq (1-\leps)\cdot \wei(G)$, we have that $p$ is well-defined and satisfies $0\leq p\leq k$.

Consider now the connected graph $G'=G[\{v_p\}\cup V(D_p)]$ and the vertex $u':=v_p$ in it.
Since $\wei(G')\geq \wei(D_p)>\wei(G)/2$, we have $\wei(N_{G'}[v])\leq \leps^{40}\cdot \wei(G)\leq \leps^{16}\cdot \wei(G')$ for each vertex $v$ of $G'$.
Hence, we can apply Lemma~\ref{lem:find-claw} to $G'$ (with the weight function $\wei(\cdot)$), vertex $u'$, and parameters $t$ and $\leps^2$.
This either yields 
\begin{enumerate}[label=(C'\arabic*),ref=(C'\arabic*)]
\item\label{c:p-claw} an induced $(\geq t)$-claw $T'$ in $G'$ with $u'$ being one of its tips; or
\item\label{c:p-esd}  a vertex subset $X'\subseteq V(G')$ and an extended strip decomposition $(H',\esd')$ of $G'-X'$ such that 
$$\wei(A)\leq (1-\leps^{14})\cdot \wei(G')\quad\textrm{and}\quad\wei(X')\leq \leps^2\cdot \wei(G'-A)\quad\textrm{for every atom }A\textrm{ of }(H',\esd').$$
\end{enumerate}
We now argue that in the second case, when conclusion~\ref{c:p-esd} is drawn, we can immediately reach conclusion~\ref{c:lesd}.

\begin{claim}
If the above application of Lemma~\ref{lem:find-claw} leads to conclusion~\ref{c:p-esd}, then conclusion~\ref{c:lesd} can be reached.
\end{claim}
\begin{proof}
Let us set
$$X = N[v_0,v_1,\ldots,v_{p-1}]\cup X'.$$
Then the graph $G-X$ is the disjoint union of $G'-X'-u'$ and all the connected components of $G_p$ different from $D_p$.
Consequently, we can obtain an extended strip decomposition $(H,\esd)$ of $G-X$ by taking $(H',\esd')$, removing $u'$ from it if $u'\notin X'$, and adding, for each component $C\in \cc(G_p)$ different from $D_p$, a new isolated vertex $x_{C}$ with $\esd(x_{C})=V(C)$. We claim that $(X,(H,\esd))$ satisfies all the properties required by conclusion~\ref{c:lesd}.

Recall that $\wei(G')\geq \wei(D_p)>(1-\leps^{35})\cdot \wei(G)$.
Take any atom $A$ of $(H,\esd)$.
If $A$ is the vertex set of a connected component $C$ of $G_p$ different from $D_p$, then we have
\begin{equation}\label{eq:otter}
\wei(A)=\wei(C)\leq \wei(G)-\wei(D_p)<\leps^{35}\cdot \wei(G)<(1-\leps^{39})\cdot\wei(G),
\end{equation}
as required. Now assume that $A$ is an atom $(H,\esd)$ that is also an atom of $(H',\esd')$ (possibly with $u'$ removed). Then by condition~\ref{c:p-esd}, we have
\begin{equation}\label{eq:squirrel}
\wei(A)\leq (1-\leps^{14})\cdot \wei(G')\leq (1-\leps^{14})\cdot\wei(G)<(1-\leps^{39})\cdot \wei(G),
\end{equation}
again as required.

Finally, let us estimate the weight of $X$. By condition~\ref{c:p-esd}, for every atom $A$ of $(H,\esd)$ that is also an atom of $(H',\esd')$ (possibly with $u'$ removed) we have
\begin{eqnarray}
\wei(X) & \leq & \wei(N[v_0,v_1,\ldots,v_{p-1}])+\wei(X')\nonumber\\
        & \leq & (\wei(G)-\wei(D_p))+\leps^2\cdot \wei(G'-A)\nonumber\\
        & \leq & \leps^{35}\cdot \wei(G)+\leps^2\cdot\wei(G-A).\label{eq:beaver}
\end{eqnarray}
On the other hand, by~\eqref{eq:squirrel} we have
$$\wei(G-A)=\wei(G)-\wei(A)\geq \leps^{14}\cdot \wei(G).$$
The above two inequalities together imply that
$$\wei(X)\leq \leps^{21}\cdot \wei(G-A)+\leps^2\cdot \wei(G-A)\leq \leps\cdot \wei(G-A).$$
This establishes the property required in conclusion~\ref{c:lesd} for atoms $A$ of $(H,\esd)$ that are actually atoms of $(H',\esd')$, possibly with $u'$ removed.
It remains to verify this property for the other atoms, that is, for connected components of $G_p$ different from $D_p$.
Let then $C$ be such a component; then by~\eqref{eq:otter} we have $\wei(C)\leq \leps^{35}\cdot \wei(G)$. Hence, by~\eqref{eq:beaver} we have
$$\wei(X)\leq \leps^{35}\cdot \wei(G)+\leps^2\cdot \wei(G')\leq 2\leps^2\cdot \wei(G)\leq \leps\cdot \wei(G-C),$$
and we are done.
\cqed\end{proof}

We continue the proof of Claim~\ref{cl:right-claw}:
we are left with considering what happens in case conclusion~\ref{c:p-claw} is drawn as a consequence of applying Lemma~\ref{lem:find-claw}.
Let $c$ be the center of the constructed $(\geq t)$-claw $T'$ and let $T$ be the induced $t$-claw in $T'$, 
that is, $T$ the subgraph of $T'$ induced by all the vertices at distance at most $t$ from the center $c$.
We define $w$ as the tip of $T$ that lies on the path connecting $u'$ and $c$ in $T'$, and we let $R$ be the subpath of this path with endpoints $u'$ and $w$.
See Figure~\ref{fig:lobster-right-claw} for an illustration.
We claim that either we can again reach conclusion~\ref{c:lesd}, or $T$ and $w$ satisfy the properties from the statement of the claim.

\begin{figure}[htbp]
\begin{center}
\def\svgwidth{0.85\textwidth}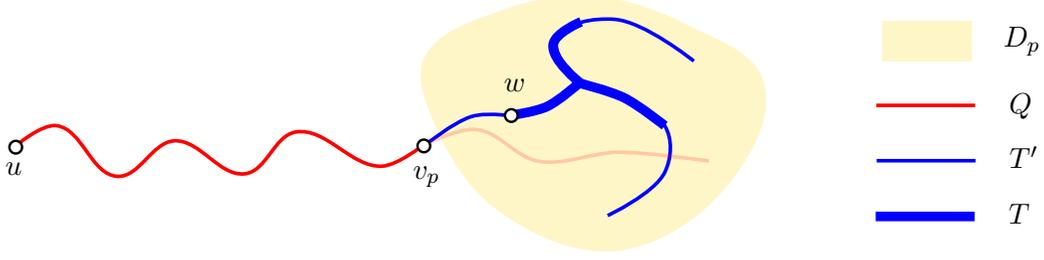
\end{center}
\caption{Situation in the proof of Claim~\ref{cl:right-claw}}\label{fig:lobster-right-claw}
\end{figure}

Let $D$ be the heaviest connected component of $G-N[V(T)\setminus \{w\}]$.
Taking $X=N[V(T)\setminus \{w\}]$, we have $\wei(X)\leq 3t\leps^{40}\cdot \wei(G)\leq \leps^{39}\cdot \wei(G)$.
Therefore, if we had $\wei(D)\leq (1-\leps^{35})\cdot \wei(G)$, then $X$ together with the trivial extended strip decomposition of $G-X$ would satisfy the requirements of conclusion~\ref{c:lesd}.
Indeed, for every connected component $D'$ of $G-X$ we would have $\wei(D')\leq \wei(D)\leq (1-\leps^{35})\cdot\wei(G)$, which entails $\wei(X)\leq \leps^{39}\cdot \wei(G)\leq \leps^{4}\cdot \wei(G-D')$.
Hence, from now on assume that $\wei(D)>(1-\leps^{35})\cdot \wei(G)$.

It remains to argue that $D$ is adjacent to $w$.
Let $\widehat{R}$ be the path obtained by concatenating the prefix of $Q$ from $u$ to $v_{p}$ with the path $R$, and removing $w$ (note that in case $w=v_p$, we also remove it from $\widehat{R}$).
Observe that $\widehat{R}$ is adjacent to $w$ and is entirely contained in $G-N[V(T)\setminus \{w\}]$, because $T'$ is an induced subdivided claw in $G'=G-(N[v_0,\ldots,v_{p-1}]\setminus \{v_p\})$.
Therefore, it suffices to argue that $N[\widehat{R}]$ and $D$ intersect.

By the choice of $p$, every connected component of $G_{p+1}=G-N[v_0,v_1,\ldots,v_p]$ has weight at most $(1-\leps^{35})\cdot \wei(G)$. On the other hand, we know that $D$ is connected in $G$ and $\wei(D)>(1-\leps^{35})\cdot \wei(G)$. Therefore, $D$ has to intersect $N[v_0,v_1,\ldots,v_p]$. We now have either $w\neq v_p$ and $N[v_0,v_1,\ldots,v_p]\subseteq N[\widehat{R}]$ and we are done,
or $w=v_p$. In the latter case, either $D$ actually intersects $N[v_0,v_1,\ldots,v_{p-1}]\subseteq N[\widehat{R}]$, or $D$ intersects $N[w]$, directly implying that $D$ is adjacent to $w$.
\cqed\end{proof}

We continue the proof of Lemma~\ref{lem:find-lobster}.
By applying Claim~\ref{cl:right-claw}, we may assume that we have constructed a suitable $t$-claw $T$ and its tip~$w$.
Let us denote by $D$ the heaviest connected component of $G-N[V(T)\setminus \{w\}]$; then Claim~\ref{cl:right-claw} ensures us that
$$\wei(D)>(1-\leps^{35})\cdot \wei(G)\qquad\textrm{and}\qquad D\textrm{ is adjacent to }w.$$
Let us define
$$G''=G[V(D)\cup \{w\}].$$
Note that $G''$ is connected.

We first verify that achieving an appropriate variant of conclusion~\ref{c:lesd} for $G''$ is sufficient for our needs.

\begin{claim}\label{cl:restrict-Gpp}
Suppose we construct a set $X''\subseteq V(G'')$ and an extended strip decomposition $(H'',\esd'')$ of $G''-X''$ with the following property:
$$\wei(A)\leq (1-\leps^{35})\cdot \wei(G'')\quad\textrm{and}\quad\wei(X'')\leq \leps^2\cdot \wei(G''-A)\quad\textrm{for every atom }A\textrm{ of }(H'',\esd'').$$
Then we can reach conclusion~\ref{c:lesd}.
\end{claim}
\begin{proof}
Set $X=X''\cup N[V(T)\setminus \{w\}]$ and observe that the graph $G-X$ can be obtained by taking a disjoint union of the graph $G''-X''-w$ and adding all the connected components of $J:=G-N[V(T)\setminus \{w\}]$
that are different from $D$.
Hence, we can construct an extended strip decomposition $(H,\esd)$ of $G-X$ by taking $(H'',\esd'')$, removing $w$ if necessary, and adding, for each component $C\in \cc(J)$ different from $D$, a new isolated vertex $x_{C}$ with $\esd(x_{C})=V(C)$. We claim that $(X,(H,\esd))$ satisfies all the properties required by conclusion~\ref{c:lesd}.

Recall that $\wei(J)\geq \wei(D)>(1-\leps^{35})\cdot \wei(G)$.
Take any atom $A$ of $(H,\esd)$.
If $A$ is the vertex set of a connected component $C$ of $J$ different from $D$, then we have
\begin{equation}\label{eq:otter2}
\wei(A)=\wei(C)\leq \wei(J)-\wei(D)<\leps^{35}\cdot\wei(G)<(1-\leps^{39})\cdot \wei(G),
\end{equation}
as required. Now assume that $A$ is an atom of $(H,\esd)$ that is also an atom of $(H'',\esd'')$ (possibly with $w$ removed). Then by the assumption of the claim we have
\begin{equation}\label{eq:squirrel2}
\wei(A)\leq (1-\leps^{35})\cdot \wei(G'')\leq (1-\leps^{39})\cdot \wei(G),
\end{equation}
again as required.

Finally, let us estimate the weight of $X$. By the assumption, for every atom $A$ of $(H,\esd)$ that is also an atom of $(H'',\esd'')$ (possibly with $w$ removed) we have
\begin{eqnarray}
\wei(X) & \leq & \wei(N[V(T)\setminus \{w\}])+\wei(X'')\nonumber\\
        & \leq & 3t\leps^{40}\cdot \wei(G)+\leps^2\cdot \wei(G''-A)\nonumber\\
        & \leq & \leps^{39}\cdot \wei(G)+\leps^2\cdot \wei(G-A).\label{eq:beaver2}
\end{eqnarray}
On the other hand, by~\eqref{eq:squirrel2} we have
$$\wei(G-A)=\wei(G)-\wei(A)\geq \leps^{35}\cdot \wei(G).$$
The above two inequalities together imply that
$$\wei(X)\leq \leps^{4}\cdot \wei(G-A)+\leps^2\cdot \wei(G-A)\leq \leps\cdot \wei(G-A).$$
This establishes the property required in conclusion~\ref{c:lesd} for atoms $A$ of $(H,\esd)$ that are actually atoms of $(H'',\esd'')$, possibly with $w$ removed.
It remains to verify this property for the other atoms, that is, for connected components of $J$ different from $D$.
Let $C$ be such a component; then by~\eqref{eq:otter2} we have $\wei(C)\leq \leps^{35}\cdot \wei(G)$. Hence, by~\eqref{eq:beaver2} we have
$$\wei(X)\leq \leps^{35}\cdot \wei(G)+\leps^2\cdot \wei(G)\leq 2\leps^2\cdot \wei(G)\leq \leps\cdot \wei(G-C),$$
and we are done.
\cqed\end{proof}

Therefore, from now on we may focus on the graph $G''$. The intuition is that $T$ is already one claw of the lobster, and in $G''$ we try to first construct the second claw, and finally the ``tail''.

Apply Lemma~\ref{lem:gyarfas-path} to the graph $G''$, vertex $w$, and $\alpha=\leps$.
This yields a suitable path $P=(y_0,y_1,y_2,\ldots,y_\ell)$, where $y_0=w$.
We adopt the notation from the statement of Lemma~\ref{lem:gyarfas-path} in the following form:
$G''_0=G''-w$ and $G''_i=G''-N[y_0,\ldots,y_{i-1}]$ for $i\in \{1,\ldots,\ell+1\}$.
Moreover, for $i\in \{0,1,\ldots,\ell+1\}$, let $D''_i$ be the heaviest connected component of $G''_i$;
then $\wei(D''_i)>(1-\leps)\cdot \wei(G'')$ for $i\leq \ell$ and $\wei(D''_{\ell+1})\leq (1-\leps)\wei(G'')$.
Again, we may assume that $\wei(D''_0)>(1-\sigma^{35})\cdot \wei(G'')$, which in particular entails $\ell\geq 0$: otherwise, the prerequisites of Claim~\ref{cl:restrict-Gpp} can be achieved by taking
$X''=\{w\}$ and the trivial extended strip decomposition of $G''-X''$, so we can reach conclusion~\ref{c:lesd}.

Let us define $p,q,r$ as the largest indices satisfying the following:
$$\wei(D''_p)>(1-\leps^{30})\cdot\wei(G)\qquad\textrm{and}\qquad\wei(D''_q)>(1-\leps^{25})\cdot\wei(G)\qquad\textrm{and}\qquad\wei(D''_r)>(1-\leps^{20})\cdot \wei(G).$$
Since $\wei(D''_0)>(1-\leps^{35})\cdot \wei(G)$ and $\wei(D''_{\ell+1})\leq (1-\leps)\cdot \wei(G)$, the indices $p,q,r$ are well-defined and satisfy $0\leq p\leq q\leq r\leq \ell$.

Similarly as in the proof of Lemma~\ref{lem:find-claw}, let us define the following subpaths of $P$:
$$R_1=(y_0,y_1,\ldots,y_{p-2}),\qquad R_2=(y_p,y_{p+1},\ldots,y_{q-2}),\qquad R_3=(y_q,y_{q+1},\ldots,y_{r-1}).$$
Observe that paths $R_1,R_2,R_3$ are pairwise disjoint and non-adjacent. 
Moreover, the same reasoning as in Claims~\ref{cl:paths-heavy} and~\ref{cl:close-done} in the proof of Lemma~\ref{lem:find-claw} easily yields the following;
we note that we verify the condition provided to Claim~\ref{cl:restrict-Gpp} in order to reach conclusion~\ref{c:lesd}.

\begin{claim}\label{cl:lpaths-heavy}
If we have
$$\wei(N[R_1])\leq \leps^{33}\cdot \wei(G'')\quad\textrm{or}\quad\wei(N[R_2])\leq \leps^{28}\cdot \wei(G'')\quad\textrm{or}\quad\wei(N[R_3])\leq \leps^{23}\cdot \wei(G''),$$
then conclusion~\ref{c:lesd} can be obtained.
In particular, if the above condition does not hold, then
$$p-0>t+1\qquad\textrm{and}\qquad q-p>t+1\qquad\textrm{and}\qquad r-q>t+1.$$
\end{claim}

Hence, from now on we assume that the condition stated in Claim~\ref{cl:lpaths-heavy} does not hold, that is:
\begin{equation}\label{eq:lpaths-heavy0}
\wei(N[R_1])>\leps^{33}\cdot \wei(G'')\quad\textrm{and}\quad\wei(N[R_2])>\leps^{28}\cdot \wei(G'')\quad\textrm{and}\quad\wei(N[R_3])>\leps^{23}\cdot \wei(G''),
\end{equation}
which in particular implies that $p>t+1$, $q>p+t+1$, and $r>q+t+1$. 
Since $\wei(G'')>\wei(G)/2$, assertion~\eqref{eq:lpaths-heavy0} in particular implies that 
\begin{equation}\label{eq:lpaths-heavy}
\wei(N[R_1])>\leps^{34}\cdot \wei(G)\quad\textrm{and}\quad\wei(N[R_2])>\leps^{29}\cdot \wei(G)\quad\textrm{and}\quad\wei(N[R_3])>\leps^{24}\cdot \wei(G).
\end{equation}

Consider now the connected graph $G'''=G''[\{y_r\}\cup V(D''_r)]$ and the vertex $u''':=y_r$ in it.
Since $\wei(G''')\geq \wei(D''_r)>\wei(G'')/2>\wei(G)/4$, we have $\wei(N_{G'''}[v])\leq \leps^{40}\cdot \wei(G)\leq \leps^{24}\cdot \wei(G''')$ for each vertex $v$ of $G'''$.
Hence, we can apply Lemma~\ref{lem:find-claw} to $G'''$ (with the weight function $\wei(\cdot)$), vertex $u'''$, and parameters $t$ and $\leps^3$.
This either yields 
\begin{enumerate}[label=(C''\arabic*),ref=(C''\arabic*)]
\item\label{c:pp-claw} an induced $(\geq t)$-claw $S'$ in $G'''$ with $u'''$ being one of its tips; or
\item\label{c:pp-esd}  a vertex subset $X'''\subseteq V(G''')$ and an extended strip decomposition $(H''',\esd''')$ of $G'''-X'''$ such that 
$$\wei(A)\leq (1-\leps^{21})\cdot \wei(G''')\quad\textrm{and}\quad\wei(X''')\leq \leps^3\cdot \wei(G'''-A)\quad\textrm{for every atom }A\textrm{ of }(H''',\esd''').$$
\end{enumerate}
We now argue that in the second case, when conclusion~\ref{c:pp-esd} is drawn, we can immediately reach conclusion~\ref{c:lesd}.

\begin{claim}
If the above application of Lemma~\ref{lem:find-claw} leads to conclusion~\ref{c:pp-esd}, then we can reach conclusion~\ref{c:lesd}.
\end{claim}
\begin{proof}
Let us define
$$X'' = N_{G''}[y_0,y_1,\ldots,y_{r-1}]\cup X'''.$$
Then the graph $G''-X''$ is the disjoint union of $G'''-X'''-u'''$ and all the connected components of $G''_r$ different from $D''_r$.
Consequently, we can obtain an extended strip decomposition $(H'',\esd'')$ of $G''-X''$ by taking $(H''',\esd''')$, removing $u'''$ from it if $u'''\notin X'''$, and adding, for each component $C\in \cc(G''_r)$ different from $D''_r$, a new isolated vertex $x_{C}$ with $\esd(x_{C})=V(C)$. We claim that $(X'',(H'',\esd''))$ satisfies the prerequisites of Claim~\ref{cl:restrict-Gpp}, which then entails conclusion~\ref{c:lesd}

Recall that $\wei(G''')\geq \wei(D''_r)>(1-\leps^{20})\cdot \wei(G'')$.
Take any atom $A$ of $(H'',\esd'')$.
If $A$ is the vertex set of a connected component $C$ of $G''_r$ different from $D''_r$, then we have
\begin{equation}\label{eq:otter3}
\wei(A)=\wei(C)\leq \wei(G'')-\wei(D_r)<\leps^{20}\cdot \wei(G'')<(1-\leps^{35})\cdot \wei(G''),
\end{equation}
as required. Now assume that $A$ is an atom $(H'',\esd'')$ that is also an atom of $(H''',\esd''')$ (possibly with $u'''$ removed). Then by condition~\ref{c:pp-esd}, we have
\begin{equation}\label{eq:squirrel3}
\wei(A)\leq (1-\leps^{21})\cdot \wei(G''')\leq (1-\leps^{21})\cdot \wei(G'')\leq (1-\leps^{35})\cdot \wei(G''),
\end{equation}
again as required.

Finally, let us estimate the weight of $X''$. By condition~\ref{c:pp-esd}, for every atom $A$ of $(H'',\esd'')$ that is also an atom of $(H''',\esd''')$ (possibly with $u'''$ removed), we have
\begin{eqnarray}
\wei(X'') & \leq & \wei(N_{G''}[y_0,y_1,\ldots,y_{r-1}])+\wei(X''')\nonumber \\
          & \leq & (\wei(G'')-\wei(D''_r))+\leps^3\cdot \wei(G'''-A)\nonumber \\
          & \leq & \leps^{20}\cdot \wei(G'')+\leps^3\cdot\wei(G''-A).\label{eq:beaver3}
\end{eqnarray}
On the other hand, by~\eqref{eq:squirrel3} we have
$$\wei(G''-A)=\wei(G'')-\wei(A)\geq \leps^{21}\cdot \wei(G'').$$
The above two inequalities together imply that
$$\wei(X'')\leq \leps^{6}\cdot \wei(G''-A)+\leps^3\cdot \wei(G''-A)\leq \leps^2\cdot \wei(G''-A).$$
This establishes the property required in conclusion~\ref{c:lesd} for atoms $A$ of $(H'',\esd'')$ that are actually atoms of $(H''',\esd''')$, possibly with $u'''$ removed.
It remains to verify this property for the other atoms, that is, for connected components of $G''_r$ different from $D''_r$.
Let then $C$ be such a component; then by~\eqref{eq:otter3} we have $\wei(C)\leq \leps^{20}\cdot \wei(G'')$. Hence, by~\eqref{eq:beaver3} we have
$$\wei(X'')\leq \leps^{20}\cdot \wei(G'')+\leps^2/2\cdot \wei(G''')\leq 3\leps^2/4\cdot \wei(G'')\leq \leps^2\cdot \wei(G''-C),$$
and we are done.
\cqed\end{proof}

Hence, from now on we may assume that the application of Lemma~\ref{lem:find-claw} leads to conclusion~\ref{c:pp-claw}.
That is, we constructed an induced $(\geq t)$-claw $S'$ in $G'''$ with $u'''$ being one of the tips.

Let $S$ be the induced $t$-claw in $S'$, that is, $S$ is induced in $S'$ by all vertices at distance at most $t$ from the center of $S'$.
Let $v$ be the tip of $S$ that is the closest in $S'$ to $u'''$.
We now define $R_3'$ as the path obtained by concatenating: the path $R_3$ (leading from $y_q$ to $y_{r-1}$) and the path within $S$ from $u'''=y_r$ to $v$.
Since $S-y_r$ is by construction contained in $G_r=G-N[y_0,\ldots,y_{r-1}]$, and $P$ is an induced path in $G''$, we infer that paths $R_1,R_2,R_3'$ are pairwise disjoint and non-adjacent.
Moreover, since $R_3$ is a subpath of $R_3'$, by~\eqref{eq:lpaths-heavy} we infer that $\wei(N[R_3'])>\leps^{24}\cdot \wei(G'')$.

\newcommand{\Gf}{G^{(4)}}
\newcommand{\Hf}{H^{(4)}}
\newcommand{\esdf}{\esd^{(4)}}

Define the following prefix of $R_2$:
$$P_2=(y_p,y_{p+1},\ldots,y_{p+t-1}).$$
We now define the graph 
$$\Gf=G-(N[V(S)\setminus \{v\}]\cup N[V(T)\setminus \{w\}]\cup (N(P_2)\setminus y_{p+t-1})).$$
Note that in $\Gf$, the path $P_2$ is preserved but becomes detached in the following sense: only the endpoint $y_{p+t-1}$ is adjacent to one vertex from the rest of the graph, namely $y_{p+t}$.
Observe that the paths $R_1,R_2,R_3'$ are also preserved in $\Gf$, and of course they are still disjoint and pairwise non-adjacent.

We now apply Lemma~\ref{lem:claw-shatter} to graph $\Gf$ with
$$Z=\{v,w,y_p\}.$$
This either yields an induced tree $U$ in $\Gf$ that contains $v,w,y_p$, or an extended strip decomposition $(\Hf,\esdf)$ of $\Gf$ that shatters $Z$.

In the first case, letting $U$ be inclusion-wise minimal subject to being connected and containing $v,w,y_p$, we observe that the set
$$V(T)\cup V(U)\cup V(S)$$
induces an $(\geq t)$-lobster in $G$. Thus, we reach conclusion~\ref{c:lobster}; see Figure~\ref{fig:lobster-left-claw}.

\begin{figure}[htbp]
\begin{center}
\def\svgwidth{0.85\textwidth}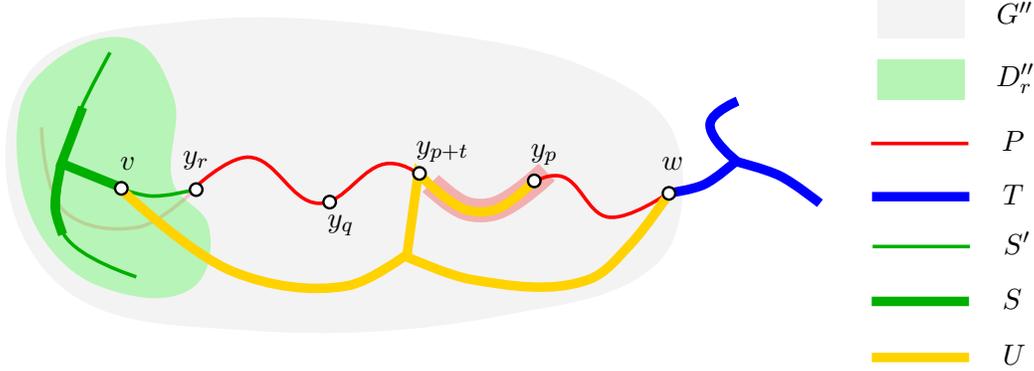
\end{center}
\caption{Final construction of the lobster}\label{fig:lobster-left-claw}
\end{figure}

We now argue that in the second case we may reach conclusion~\ref{c:lesd} by taking
$$X=N[S]\cup N[T]\cup N[P_2],$$
and an extended strip decomposition $(H,\esd)$ of $G-X$ defined as follows: 
take $(\Hf,\esdf)$ and, keeping $H=\Hf$, remove all vertices that belong to $X$ from all the sets in the image of $\esdf(\cdot)$.

Since $(\Hf,\esdf)$ shatters $Z$ in $\Gf$, while $R_1,R_2,R_3'$ are pairwise disjoint and non-adjacent paths in $\Gf$, each having an endpoint in $Z$, we infer that every atom $A$ of $(\Hf,\esdf)$
is disjoint with either $N[R_1]$, or $N[R_2]$, or $N[R_3']$. By~\eqref{eq:lpaths-heavy} we infer that $\wei(A)\leq (1-\leps^{34})\cdot \wei(G)$ for every atom $A$ of $(\Hf,\esdf)$. 
Since atoms of $(H,\esd)$ are subsets of atoms of $(\Hf,\esdf)$, we also have $\wei(A)\leq (1-\leps^{34})\cdot \wei(G)$ for every atom $A$ of $(H,\esd)$. 

Now, observe that since $|X|\leq 7t+2$, we have
$$\wei(X)\leq (7t+2)\leps^{40}\cdot \wei(G)\leq \leps^{39}\cdot \wei(G).$$
As $\wei(A)\leq (1-\leps^{34})\cdot \wei(G)$ for every atom $A$ of $(H,\esd)$, we also have $\wei(G-A)\geq \leps^{34}\cdot \wei(G)$, which in conjunction with the above yields that
$$\wei(X)\leq \leps^{5}\cdot \wei(G-A)\qquad\textrm{for every atom }A\textrm{ of }(H,\esd).$$
This means that we have indeed reached conclusion~\ref{c:lesd}.

\medskip

For the enumeration statement, if suffices to examine the consecutive steps of the reasoning and replace all steps where we invoke the existential statements of Lemmas~\ref{lem:gyarfas-path} and~\ref{lem:find-claw} with iteration over the families obtained by respective enumeration statements. The final family $\Nn$ consists of all the pairs $(X,(H,\esd))$ that we might have obtained at any point in the reasoning as witnesses for conclusion~\ref{c:lesd}, for all possible choices of objects from the families provided by Lemmas~\ref{lem:gyarfas-path} and~\ref{lem:find-claw}.
To be more precise, we first invoked Lemma~\ref{lem:gyarfas-path} followed by Lemma~\ref{lem:find-claw} in the proof of Claim~\ref{cl:right-claw}, which results in either finding an induced $t$-claw $T$ or a suitable family $\Nn$ of size $\Oh(|V(G)|^6)$. Then we again invoked Lemma~\ref{lem:gyarfas-path} followed by Lemma~\ref{lem:find-claw} in the remainder of the proof, which again results in either finding an induced $(\geq t)$-lobster or a suitable family $\Nn$ of size $\Oh(|V(G)|^6)$.
\end{proof}

\bibliographystyle{plain}

\bibliography{../references}

\end{document}